\documentclass[journal,onecolumn,12pt]{IEEEtran}

\usepackage[paperheight=28cm,paperwidth=21.5cm,textwidth=16.5cm,textheight=23cm]{geometry}
\usepackage{setspace}
\usepackage[table,usenames,dvipsnames]{xcolor}
\definecolor{Mahogany}{rgb}{0.75, 0.25, 0.0}
\usepackage{hyperref}
\hypersetup{
    unicode=false,          % non-Latin characters in Acrobat’s bookmarks
    pdftoolbar=true,        % show Acrobat’s toolbar?
    pdfmenubar=true,        % show Acrobat’s menu?
    pdffitwindow=false,     % window fit to page when opened
    pdfstartview={FitH},    % fits the width of the page to the window
    pdftitle={PDF PROJECT}, % title
    pdfauthor={Valentin Savin}, % author
    pdfsubject={Error correction for quantum information processing},   % subject of the document
    pdfcreator={Creator},   % creator of the document
    pdfproducer={Producer}, % producer of the document
    pdfkeywords={keyword1} {key2} {key3}, % list of keywords
    pdfnewwindow=true,      % links in new PDF window
    colorlinks=true,        % false: boxed links; true: colored links
    linkcolor=Mahogany,     % color of internal links
    citecolor=Mahogany,     % color of links to bibliography
    filecolor=Mahogany,     % color of file links
    urlcolor=Mahogany       % color of external links
}
\usepackage{svg}
\usepackage{amsmath,amssymb,amsthm,amsfonts}
\usepackage{subcaption}
\usepackage{graphicx}
\usepackage{textcomp, gensymb}
\usepackage{cite}
\usepackage{lipsum}
\usepackage{tikz}

\usepackage{thmtools}
\usepackage{mathtools}
\usepackage{thm-restate}
\usepackage{cleveref}
\usepackage{xspace}

\newtheorem*{theorem*}{Theorem}
\newtheorem{lemma}{Lemma}
\newtheorem{corollary}{Corollary}
\newtheorem{definition}{Definition}

\newcommand{\Cxc}{{\Cx^c}}
\newcommand{\Coc}{{\Co^c}}

\newcommand{\exgrata}{\textit{e.g.}}
\newcommand{\idest}{\textit{i.e.}}
\newcommand{\Hmat}{\mathbf{H}}
\newcommand{\evec}{\mathbf{e}}
\newcommand{\synd}{\mathbf{s}}
\newcommand{\edgeq}{e_q}
\renewcommand{\sc}{\synd^c}
\newcommand{\aposteriori}{\textit{a posteriori}\xspace}
\newcommand{\apriori}{\textit{a priori}\xspace}

\newcommand{\Cx}{C_\bullet}
\newcommand{\Co}{C_\circ}

\newcommand{\Tiq}{\mathcal{T}_{i,q}}

\newcommand{\hTiq}{\hat{\mathcal{T}}_{i,q}}

\newcommand{\bottom}{\_}

\newcommand{\link}{{\;\setlength{\fboxsep}{0pt}\protect\framebox[1.75eX]{\rule{0ex}{1ex}}\;}}
\newcommand{\xlink}{{\;\setlength{\fboxsep}{0pt}\protect\framebox[1.75eX]{\rule{1.75ex}{1ex}}\;}}

\newcommand{\app}{\textsc{app}}
\DeclareMathOperator{\supp}{supp}

\begin{document}
\title{\vspace*{-15mm}A Blindness Property of the Min-Sum Decoding for the Toric Code}
\author{{\large Julien du Crest$^1$, \qquad Mehdi Mhalla$^2$, \qquad Valentin Savin$^3$}\\
{\small  $^1$Université Grenoble Alpes, Grenoble INP, LIG, F-38000 Grenoble, France}\\[-2mm]
{\small  $^2$Université Grenoble Alpes, CNRS, Grenoble INP, LIG, F-38000 Grenoble, France}\\[-2mm]
{\small  $^3$Université Grenoble Alpes, CEA-Léti, F-38054 Grenoble, France}}
\date{\vspace*{-2em}}

\maketitle

\begin{abstract}
\small\noindent Kitaev’s toric code is one of the most prominent models for fault-tolerant quantum computation, currently regarded as the leading solution for connectivity constrained quantum technologies. Significant effort has been recently devoted to improving the error correction performance of the toric code under message-passing decoding, a class of low-complexity, iterative decoding algorithms that play a central role in both theory and practice of classical low-density parity-check codes. Here, we provide a theoretical analysis of the toric code under min-sum (MS) decoding, a message-passing decoding algorithm known to solve the maximum-likelihood decoding problem in a localized manner, for codes defined by acyclic graphs. Our analysis reveals an intrinsic limitation of the toric code, which confines the propagation of local information during the message-passing process. We show that if the unsatisfied checks of an error syndrome are at distance $\ge 5$ from each other, then the MS decoding is locally blind: the qubits in the direct neighborhood of an unsatisfied check are never aware of any other unsatisfied checks, except their direct neighbor.  Moreover, we show that degeneracy is not the only cause of decoding failures for errors of weight at least 4, that is, the MS non-degenerate decoding radius is equal to $3$, for any toric code of  distance $\ge 9$. Finally, complementing our theoretical analysis, we present a pre-processing method of practical relevance. The proposed method, referred to as stabiliser-blowup, has linear complexity and allows correcting all (degenerate) errors of weight up to 3, providing quadratic improvement in the logical error rate performance, as compared to MS only.  
\end{abstract}

\section{Introduction}

 Message passing (MP) decoding stands as both the fundamental theoretical reason behind the invention of classical LDPC codes (departing from previous minimum distance centric constructions), and the key to their success and widespread application in real-world systems. However, quantum LDPC codes are known to be classically degenerate~\cite{poulin2008iterative},  thus the sparsity of their Tanner graph does no longer fulfill its role of enabling efficient MP decoding, but merely acts as an enabler for fault-tolerance (\emph{\exgrata}, fault-tolerant syndrome extraction and fault-tolerant operations on logical qubits). Consequently, a significant effort has been devoted over the last few years to efficient decoding of quantum LDPC codes, by either combining the MP decoding with a post-processing step~\cite{panteleev2021degenerate, roffe2020decoding, ducrest2022stabilizer, ducrest2024check}, and/or improving the MP decoding performance itself~\cite{wang2012enhanced, rigby2019modified, kuo2022exploiting, yao2023belief, old2023generalized}. 

Here, we consider the MP decoding of the Kitaev toric code~\cite{kitaev2003fault}, the planar version of which is currently  the dominant error-correction solution to achieve fault-tolerance in large-scale quantum computers, especially for connectivity constrained  technologies~\cite{google2023suppressing}. Among the large class of quantum LDPC codes, the toric code is also known to be the less responsive to  MP, due the the presence of weight-$2$ degenerate errors undecodable by vanilla MP decoders, such as belief-propagation (BP)\footnote{We use BP for the MP decoding  using Bayesian  rules to update the exchanged messages (also known as sum-product). It performs the  exact Bayesian inference on trees, but may be  used on more general graphical models. Note that BP is sometimes used (included in some of the papers cited here) with the generic meaning of MP.} or min-sum (MS).  Consequently, the logical error rate of these vanilla decoders scales as $p^2$, where $p$ is the physical error rate, well behind the $p^{\left\lceil \frac{d}{2} \right\rceil}$ scaling of a minimum distance decoding (\emph{\idest}, correcting any error of weight  $\leq \left\lfloor \frac{d-1}{2} \right\rfloor$), where $d$ is the minimum distance of the code. A few approaches proposed in the literature, such as   BP with memory~\cite{kuo2022exploiting}, generalized BP~\cite{old2023generalized}, or neural-BP~\cite{NNBP_QEC, varsamopoulos2019comparing}, succeeded to improve the logical error rate performance for small distance toric codes ($d\leq 9$), but they all failed to scale to larger distance codes. Precisely, for $d\leq 9$, the above approaches may yield a logical error rate that scales as $p^{\left\lceil \frac{d}{2} \right\rceil}$, but this scaling does not longer improve with increasing minimum distance\footnote{The Adaptive BP with Memory (AMBP) in~\cite{kuo2022exploiting} may improve the slope of the logical error rate curve beyond $d=9$, which comes from the use of many MBP decoders (with different meta-parameter values). Put differently, this improvement is attributable to a form of decoding diversity, rather than the intrinsic error correction capacity of the MBP decoder, which saturates at $d=9$ \cite[Section~4.4]{kuo2022exploiting}.}, indicating the presence of undecodable errors of weight $5$. Here, we try to understand if such difficulties  come from an intrinsic limitation of the toric code itself, hindering further improvement of MP-based decoding for $d>9$.   

We consider binary MP decoding, corresponding to separate decoding of $X$ and $Z$-type errors. The input of the decoder consists of the error syndrome and  the \apriori \emph{soft value} of each qubit (\emph{\exgrata}, error probabilities or log-likelihood ratio values). Decoding is carried out through an iterative exchange of extrinsic messages between qubits and checks they participate in, which is conveniently described via the  Tanner graph representation of the code. In this representation, messages are exchanged between neighboring qubit and checks of the Tanner graph, and the iterative nature of the decoding process is aimed at spreading the information globally.
At each decoding iteration, an \aposteriori \emph{soft value}is computed for each qubit, based on the collected messages, gradually improving the estimation of the  error. In the sequel, we shall simply use the terminology of  \apriori and \aposteriori value (omitting soft). For  details on MP decoders we refer to~\cite{Wiberg96}.

Our goal is to understand how the information is spreading on the toric code, and what is the maximum length it can travel. To provide a formal statement, we first introduce the notion of local blindness of  an MP decoder. Informally, given an error syndrome $\synd$ and an unsatisfied check $c$, \emph{\idest}, such that $\synd(c)=1$, we consider a fake syndrome $\sc$, having $c$ as the only unsatisfied check (note that $\sc$ is not  a valid error syndrome, since no error can generate it). We say that the MP decoding of $\synd$ is \emph{locally blind} in the neighborhood of $c$, if running the MP decoding on the error syndrome $\synd$, or on the fake syndrome $\sc$, yields the same \aposteriori 
 value for any qubit $q$ neighboring  $c$, and for any number of decoding iterations. Put differently,  at no iteration are the qubits  neighboring $c$  aware that it is not the only unsatisfied check of the syndrome: the MP decoder fails to convey the information from the other unsatisfied checks. Intuitively, this may happen when $c$ is too far away from the other unsatisfied checks of the syndrome. 
 
 \subsection*{Main Results}

In this work, we focus on the MS decoding, with the conventional flooded scheduling. The reason is twofold. First, MS is an MP decoding algorithm that is aimed at solving the maximum likelihood (ML) decoding problem in a \emph{localized} manner (see Section~\ref{sec:preliminaries} for details). It actually succeeds in doing so for  codes defined by acyclic graphs.  However, quantum LDPC codes (as well as good classical LDPC codes) are defined by graphs with cycles,  preventing local information from  being spread effectively, and causing a degradation of the error correction performance with respect to ML decoding. Second,  MS  presents a number of practical advantages, \emph{\exgrata}, low computational complexity (only requires additions and comparisons) and robustness to low precision arithmetic (\emph{\exgrata}, below $6$-bit messages), being the \emph{de facto} solution used in practical applications and hardware implementations~\cite{boutillon2014hardware}. Moreover, assuming a Pauli channel model for physical qubit errors, with $X$ and $Z$ errors  decoded independently, MS does not need an \apriori knowledge of the channel $X, Y, Z$ error probabilities~\cite{ducrest2022stabilizer}. 
 
 Our first result states that the \textbf{local} information exchange of the MS is not spreading \textbf{globally}:

% The MS Limited Scope
\begin{restatable}[MS Local Blindness]{thm}{limitedscope}\label{thm:limitedscope}
 Consider an error syndrome $\synd$ on a toric code, such that all the unsatisfied checks are at distance at least 5 from each other. 
 Then, under the i.i.d. noise assumption,  the MS decoding of $\synd$ is locally blind in the neighborhood of any unsatisfied check. 
\end{restatable}
In particular, it follows that such a syndrome is undecodable by  MS. Further, we are interested in determining the smallest weight of an undecodable non-degenerate error. We say that an error $\evec$ is \textbf{non-degenerate} if it generates a syndrome that cannot be generated by any other error $\evec'\neq \evec$, such that $|\evec'|\leq |\evec|$. For instance, 1-dimensional errors $\evec$ of weight $|\evec| \leq \left\lfloor \frac{d-1}{2}\right\rfloor$ and generating a syndrome of weight $2$ are non-degenerate. But some 2-dimensional errors may also be non-degenerate.

Before introducing the next Theorem, we refine the classical notion of decoding radius (\idest, the largest integer $\omega$, such that any   error of weight $\leq \omega$ is decoded correctly), by introducing the \textbf{non-degenerate decoding radius}, which only takes into account non-degenerate errors.

\begin{restatable}[MS Non-Degenerate Decoding Radius] {thm}{weightfourbis}\label{thm:4inarow}
    For any toric code of distance $\ge 9$, the non-degenerate decoding radius of the MS is 3. 
\end{restatable}

We give formal proofs of Theorems~\ref{thm:limitedscope} and~\ref{thm:4inarow}, based on the formalism of decoding trees. 
We also conjecture that Theorem \ref{thm:limitedscope}  holds for the normalized MS decoder, irrespective of the normalization factor (note also that for an appropriate choice of the normalization factor, the normalized MS provides a good approximation of BP). 
Although the proof techniques we use do not translate directly to normalized MS, we provide numerical evidence to corroborate our conjecture.

Finally, we propose a method to remove low-weight degeneracy on the toric code, which naturally complements the above analysis. 
Although the toric code is know to be the hypergraph product code of two classical repetition codes~\cite{tillich2013quantum}, on which the MS decoding radius is equal to $\left\lfloor \frac{d-1}{2}\right\rfloor$, the decoding radius of the MS on the toric code is only equal to~1, and the essence of Theorem~\ref{thm:4inarow} is that even the \textbf{non-degenerate} decoding radius is no greater than~$3$.
We propose a new pre-processing that runs in linear time and allows  correcting all (degenerate) errors of weight up to~$3$ thus improving the decoding radius of the MS to 3. We refer to this pre-processing step as stabiliser-blowup (SB), and it amounts to applying a change of variable that locally removes the degeneracy to allow the decoder to converge. We formally prove the following Theorem.

\begin{restatable}[Stabilizer Blowup Pre-Processing]{thm}{SBMS}\label{thm:SBMS}
    SB+MS is able to correct all errors of weight up to 3 on a toric code of distance $\ge 7$. 
\end{restatable}

It follows that the SB+MS yields a quadratic improvement in the logical error rate performance, as compared to MS only. In particular, it can be applied to quadratically reduce the number of calls to a possible post-processing step, which remains a must\footnote{According to the limited performances of state of the art MS improvements, of which the poor performances are likely being explained  by the local blindness property from Theorem~\ref{thm:limitedscope}.} for the MS-based decoding of toric codes of  minimum distance $d \ge 9$.

\section{Preliminaries}
\label{sec:preliminaries}

\subsection{The Decoding Problem}
\label{subsec:decoding-problem}

In the rest of the paper, we will restrict ourselves to a simple setup: decoding independent and identically distributed (i.i.d.) $X$ errors on the toric code. In this noise model, the \apriori values of all qubits are the same, and only depend on the probability of $X$ error on each qubit. Everything said here also applies to $Z$ errors.

Since the toric code is a CSS code, the maximum likelihood\footnote{Here we only discuss the ML decoding of quantum codes and not the quantum ML decoding, where one is interested in decoding the most likely equivalence class of errors. This simplification of the problem makes sense in the rest of the paper since our goal is to show that the decoder is unable to decode even constant sized errors, so we are not so much concerned about logical errors, but rather about decoding failures (inability to output a valid correction).} (ML)  decoding of $X$ errors can be stated as: for a matrix $\Hmat$, and a syndrome $\synd$, find the most likely error $\tilde{\evec}$ such that $\Hmat\tilde{\evec}=\synd$. Assuming an i.i.d. noise model, finding the most likely error is equivalent to finding the error of smallest weight.

A tie occur in the ML decoding if there are several smallest weight  (or in the general setting, most likely) errors satisfying the syndrome. If the decision is taken globally,  the ML decoder may output either the list of all such errors, or any one of them (\exgrata, randomly chosen). This is no longer possible for decoders based on local decisions, such as MP decoders, where each qubit is decided based on local information. In this case, it is more convenient to  consider the following \emph{localized} version of the ML decoder. For a matrix $\Hmat$, a syndrome $\synd$, and a qubit $q$: if all the smallest weight errors satisfying the syndrome agree on qubit $q$, that is, $\exists a\in\{0,1\}$, such that $\evec(q) = a$ for any smallest weight $\evec$ satisfying $\Hmat\evec = \synd$, then set $\tilde{\evec}(q) = a$; otherwise, set $\tilde{\evec}(q)$ to a random value in $\{0,1\}$. Obviously, if there is only one smallest weight error satisfying the syndrome, then both the ML and its localized variant output this error. Otherwise, the error outputted by the localized ML may actually not satisfy the given syndrome. 

\subsection{Degenerate Errors}

For a given syndrome $\synd$, we write $\synd(c)$ to denote the value of $\synd$ on check $c$.
If $\synd(c) = 1$, we say that check $c$ is \textbf{unsatisfied}, otherwise we say it is \textbf{satisfied}.

The notions of degenerate syndrome and degenerate error, introduced below, capture the fact that  there may be several errors of minimal weight explaining a given syndrome. Such a syndrome would intuitively be hard to decode, since it would be difficult for the decoder to arbitrate between those errors, as it is a classical decoder, not aware of the degeneracy of the code.

For a given parity-check matrix $\Hmat$, for any syndrome $\synd$, we denote by $\mathcal{E}(\synd)$ the set of errors satisfying this syndrome. $$\mathcal{E}(\synd) =  \{ \evec \enskip |  \Hmat\evec=\synd\}$$ 

Moreover, we denote by $|\evec|$  the (Hamming) weight of an error $\evec$, and  define $\mathcal{E}_{\min}(\synd)$ as the subset of errors of minimal weight. 
$$\mathcal{E}_{\min}(\synd) = \{ \evec \in \mathcal{E}(\synd) \mid |\evec| = \min_{\evec'\in \mathcal{E}(\synd)}  |\evec'|\}$$

\begin{definition}
    A syndrome $\synd$ is said \textbf{degenerate} if $|\mathcal{E}_{\min}(\synd)| > 1$.
\end{definition}
\begin{definition}
    
    An error $\evec$ is said \textbf{degenerate} if there exists $\evec' \neq \evec$ such that $|\evec'|\leq |\evec|$ and $\Hmat\evec' = \Hmat\evec$.
\end{definition}

\begin{lemma}
    If $\evec \in \mathcal{E}_{\min} (\synd)$, then $\evec$ is degenerate iff $\synd$ is degenerate.
\end{lemma}

We also refine the classical notion of decoding radius to discuss more precisely the decoding of non-degenerate errors.

\begin{definition}
    The \textbf{non-degenerate} decoding radius $\omega$, for a given decoder and code, is the largest integer $\omega$ such that all non-degenerate errors of weight $\le \omega$ are correctly decoded. 
\end{definition}

\subsection{The MS Decoder}

For LDPC codes, the two main MP decoders used are the BP and the MS. On acyclic graphs, the first solves the maximum \aposteriori decoding problem, while the second solves the \emph{localized} ML decoding problem, discussed above (see also~\cite{Wiberg96}). Algorithmically, they both function by iteratively exchanging messages on the Tanner graph of the code. They are usually described in an algorithmic way, by specifying the functions used to compute the exchanged messages, \textcolor{black}{and they can be analysed by means of decoding (or computation) trees, introduced by Wiberg~\cite{Wiberg96}. }
We refer to~\cite{Wiberg96} for the algorithmic description of the MS and the definition of the decoding tree in the general setting, \textcolor{black}{and provide below the definition of the decoding tree for the particular case of the toric code.}

It is straightforward to see that the localized ML corrects all the non-degenerate errors, and fails (with high probability\footnote{Since ties are resolved randomly in the localized ML, there is a non-zero probability that such random choices yield a  valid minimum-weight explanation of the syndrome.}) on the degenerate ones. Although the MS decoder aims at solving the localized ML decoding problem, we show in this paper that it fails even to decode some non-degenerate errors of small weight on the toric code. However, counter intuitively, the MS is also able to correct some degenerate errors. A discussion on the decoding of degenerate errors by the MS on the toric code can be found in Appendix~\ref{app:add-res-degen}.

\subsection{The Decoding Tree of the Toric Code}

\textcolor{black}{From now on, all the definitions and results presented apply to the toric code, even when this is not specified explicitly.}
We represent the toric code as a square tiling of the torus, \idest, a two-dimensional square lattice with periodic boundaries,  where the qubits are located on the edges, the $Z$ checks on the vertices, and the $X$ checks or the plaquettes. Note that the definition of $X$ and $Z$ checks is reversed by considering the dual lattice. Thus, we shall only consider here one type of checks, and we assume they are located on the vertices of the lattice. A toric code is depicted in Figure~\ref{fig:toric_tiling}  (or see \cite{kitaev2003fault} for a formal introduction). 

For the graph representation of the toric code, and more generally for any graph $G=(V,E)$, we use the notation $\mathcal{N}(\enskip)$ to refer to the set of edges incident to a vertex, or the set of vertices incident to an edge, clear from the context:
$$\mathcal{N}(v) = \{e\in E \mid e \text{ incident to } v\}, \qquad \mathcal{N}(e) = \{v\in V \mid v \text{ incident to } e\}$$

As mentioned above, the MS algorithm converges to a localized ML decision if the Tanner graph of the code is a tree~\cite{Wiberg96}. This is not true if the Tanner graph contains cycles, however, there is a neat way to explicitly compute the \aposteriori  value that the MS will output for every qubit at every iteration, using the notion of minimal configurations on decoding trees. Those ideas where first presented in \cite{Wiberg96}, and we quickly review here what will be necessary for the proofs of the next section. 

\begin{figure}[!th]
\centering
\begin{subfigure}{0.5\textwidth}
    \includegraphics{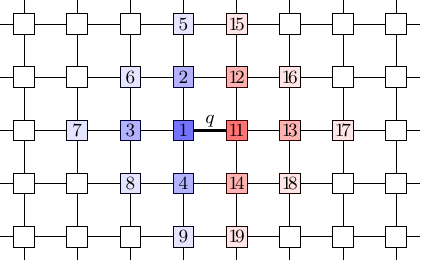}
    \caption{Square tiling of the torus}
    \label{fig:toric_tiling}
\end{subfigure}

\begin{subfigure}{\textwidth}
\centering
    \includegraphics{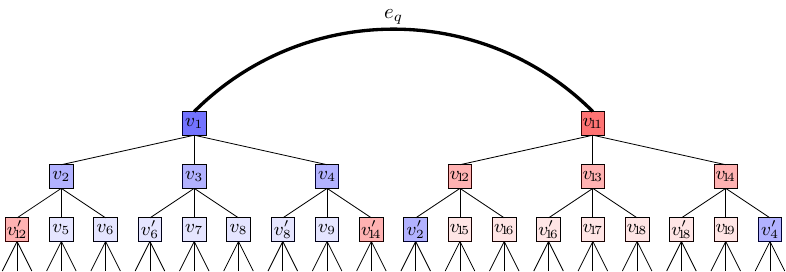}
    \caption{Decoding tree $\Tiq$ of qubit $q$ at iteration $i=3$}
    \label{fig:toric_dectree}
\end{subfigure}
\caption{Decoding tree of the toric code. (a) Toric code defined by a square lattice with periodic boundaries. Checks are represented by squares, placed on the lattice vertices, and qubits are represented by  edges.  (b) The decoding tree of qubit $q$  is shown for three decoding iterations.  A check $c$ from the original lattice may have several associated vertices in the decoding tree, denoted as $v_c$, $v'_c$. Similarly, each qubit from the original lattice may have several associated edges in the decoding tree (\exgrata, the edge $(v_1,v_2)$ and one of the dangling edges of $v'_2$ are both associated with the qubit incident to checks $1$ and $2$ in the original lattice). We consider the decoding tree as being rooted in the edge $e_q$, associated with  qubit $q$ from the original lattice (see main text).  
}
\label{fig:dectree}
\end{figure}

For a given qubit $q$, one may recursively trace back through time the computation of its \aposteriori  value, by examining the updates that have occurred. This trace back will form a tree graph rooted at $q$ and consisting of interconnected qubit and checks in the same way as in the original graph, but the same qubit or checks may appear at several places in the tree because of the loopy structure of the decoding graph. For toric codes, checks are represented by vertices of a square lattice, while qubits are represented by edges (Figure~\ref{fig:toric_tiling}). Accordingly,  the decoding tree of the toric code   will have vertices and edges representing, respectively, checks and qubits at different decoding iterations.  Decoding trees can be formally defined by using the notion of walk~without~return, introduced below.

\begin{definition}
    A \textbf{walk} $W$ in a graph $G = (V,E)$ as a sequence of edges $W = e_1 e_2\dots e_k$, for which there exists a sequence of vertices $v_1 v_2\dots v_{k+1}$, such that
    $$
        e_i = (v_i, v_{i+1}), \quad \forall i\in\{1, k\}.
    $$
 We say that $W$ is a \textbf{walk~without~return~(wwr)} if $\enskip v_{i} \neq v_{i+2}, \enskip \forall i \in \{1,\ldots, k-1\}$. The length $|W|$ of a walk is the number of edges it contains.
\end{definition}

\begin{definition}
    A \textbf{path} is a walk such that all edges are distinct.
\end{definition}

\begin{definition}\label{def:decoding-tree} (For the toric code)
The \textbf{decoding tree} $\Tiq$ is composed of  all the wwr of length $i$ in the toric code lattice, starting with edge $\edgeq$ (associated to qubit $q$).\smallskip

When necessary, we also subscript the vertices with the name of the associated check, so $u_c \in V(\Tiq)$ is a vertex in $\Tiq$ associated to check $c$.
\end{definition}

A decoding tree is illustrated in Figure~\ref{fig:toric_dectree}. Note that we consider this tree as being rooted in the edge $e_q$, and that it ends with \textbf{dangling edges} (defined below). Technically, we could insert a qubit vertex on each edge, such that the decoding tree would be rooted in a qubit vertex, and terminated with qubit vertices. We will not do this here, and will  rather stick to the usual convention for toric codes, where qubits are represented by edges. Remark that since a decoding tree is defined by walks without return, each vertex $u$ has exactly 3 children (see Figure~\ref{fig:toric_dectree}).

\begin{definition}
    A \textbf{dangling edge} of $\Tiq$ is an edge $e$ such that $|\mathcal{N}(e)| = 1$.
\end{definition}

In the context of decoding trees, we naturally extend the definitions of \textbf{walk}, \textbf{walk~without~return} and \textbf{path}, that can now eventually start/end with dangling edges.

\begin{definition}
    The distance between two vertices in  $\Tiq$ is the length of the path (the number of edges) from one to the other. 
    The distance of a vertex to the bottom is the minimum over the length of all paths that start with that vertex and end with a dangling edge.
\end{definition}

Since the graph representation of the toric code contains loops, several vertices/edges in the decoding tree may be associated to the same check/qubit of the toric code. 
Therefore, we extend the notation of the syndrome to vertices of $\Tiq$ as: $\synd(v_c) = \synd(c),  \forall v_c \in V(\Tiq)$.

\begin{definition}\label{def:config}
For a decoding tree $\Tiq$ and a syndrome $\synd$, a \textbf{configuration} is a function  $C:E(\Tiq)\rightarrow \{0,1\}$ such that:
$$  \sum_{e \in \mathcal{N}(v)} C(e)   = \synd(v) \bmod 2, \enskip \forall v \in V(\Tiq)$$
If $C(e)=1$, the edge $e$ is said to be \textbf{labeled} otherwise it is \textbf{unlabeled}.

Additionally, we are interested to know the labeling of the root edge of the tree, and will call a \textbf{root-labeled} configuration $\Cx$ if $ \Cx(\edgeq) = 1$ and a \textbf{root-unlabeled} configuration $\Co$ if $\Co(\edgeq) = 0$.
\end{definition}

\begin{definition}
We define the \textbf{ weight } of a configuration as:
$$|C| = \sum_{e\in E(\Tiq)} C(e) \quad \in \mathbb{N}$$
\end{definition}

 Let $\mathfrak{C}$ be the set of all configurations for decoding tree $\Tiq$ and syndrome $\synd$. We write
$$\mathfrak{C} = \mathfrak{C}_\bullet \cup \mathfrak{C}_\circ$$ where $\mathfrak{C}_\bullet $ and $ \mathfrak{C}_\circ$  are the subset of root-labeled/unlabeled configurations. Theses sets are non-empty since it is always possible to define a configuration starting from the root down to the dangling edges.

\begin{definition}
A \textbf{minimal root-labeled/unlabeled configuration} $\Co^*$/$\Cx^*$ is a configuration such that 
$$|\Co^*| = \min_{\Co \in \mathfrak{C}_\circ} |C| \quad / \quad |\Cx^*| = \min_{\Cx \in \mathfrak{C}_\bullet} |C|$$
\end{definition}

For i.i.d. noise, all qubits have the same \apriori value, which factors out through all the MS decoding steps, and thus can be  set  to 1 (which  amounts to saying that the maximum likelihood error is the minimum weight one). This will be implicitly assumed in the next theorem and throughout the rest of the paper.

\begin{restatable}[\!{\cite[Section 4.1]{Wiberg96}}\,]{thm}{iid}\label{theo:Wiberg96}
For i.i.d. noise, and given minimal configurations $\Co^*$ and $\Cx^*$ for the decoding tree $\Tiq$ and syndrome $\synd$, 
the \aposteriori  value $\app(q,i)$ computed by the MS for qubit $q$  at iteration $i$ is given by $\app(q,i) = |\Cx^* |- |\Co^*|$. 
\end{restatable}

\section{The MS Limited Scope}\label{sec:limited-scope}

In this section we prove Theorems \ref{thm:limitedscope} and~\ref{thm:4inarow}.
We first prove the first Theorem that introduces all the ideas and is quite short and intuitive.
Building on the ideas of the first Theorem, together with some technical Lemmas, we then show the second Theorem.

\subsection{Theorem~\ref{thm:limitedscope}: MS Local Blindness}
\label{subsec:local-blindness}

 Consider an error syndrome $\synd$, and let $c$ be an unsatisfied check, \emph{\idest}, such that $\synd(c)=1$.  Let also $\sc$ denote the binary vector, such that 
 $\sc(c')= \begin{cases} 
 1, &   \text{if $c=c'$}\\[-1ex]
 0, & \text{otherwise}
 \end{cases}
 $.
  
 Note that $\sc$ is not a valid error syndrome, as no error can generate it. If one runs  MS or any other MP decoder on $\sc$, the decoder will run for infinitely many iterations, without finding a valid explanation of the input ``fake'' syndrome $\sc$. We are interested in the \aposteriori   values computed by the  decoder for the qubits neighboring $c$. We denote by  $\app(q, i)$  the \aposteriori  value  of a qubit $q$ at iteration $i$, when the decoder is supplied with the real error syndrome $\synd$. Likewise, we denote by  $\app^c(q, i)$  the \aposteriori  value  of a qubit $q$ at iteration $i$, when it is supplied with the fake error syndrome $\sc$.
 
 \begin{definition}
Consider the MP decoding of a syndrome $\synd$, and let $c$ be un unsatisfied check. We say that \textbf{the decoding of $\synd$ is locally blind in the neighborhood of $c$}, if $\app(q, i) = \app^c(q, i)$, for any qubit $q$ neighboring $c$, and any iteration  $i \ge 0$. 

 \end{definition}

 We restate here our first Theorem.
 \limitedscope*

We will prove Theorem~\ref{thm:limitedscope} by investigating the minimal configurations in the decoding tree.  We  introduce first several  definitions and prove some preliminaries Lemmas, and delay the proof of Theorem~\ref{thm:limitedscope} to Section~\ref{sec:proof-thm-MS-local-blindness}.

\subsubsection{Links and Alternating Chains}

Below, we give a few definitions, and adapt some usual definitions in graph theory to cope with the  specificity of decoding trees (taking into account dangling edges).

\begin{definition}
A \textbf{link} is a path in the decoding tree such that its vertex endpoints (if they exist) are associated  with unsatisfied checks.

A link can have dangling edges as endpoints. 

\smallskip

We call a \textbf{proper link} a link with both endpoints associated to unsatisfied checks, and a \textbf{dangling link} a link where at least one of its endpoints is a dangling edge. 
For short, we usually only give the  vertex endpoints of a link. Giving the endpoints of a proper link is enough to fully characterize it, so we write  $u \link v$. For a dangling link, we write $ u \link \bottom$. Although this is not enough to uniquely determine a dangling link (there exist many dangling links with endpoint $u$), whenever this notation is used, it is always considered that this is one of the possible dangling links of minimal length. 

\end{definition}

\begin{definition}
Below we give a name to the 3 possible shapes of proper links of length 4 not containing the root (see Figure~\ref{fig:4links-shapes}): 

\begin{itemize}
    \item The $I$-link going from an unsatisfied check 4 times down,
    \item The $\Gamma$-link going once up then 3 times down,
    \item The $\Lambda$-link going twice up then twice down.
\end{itemize}
For the $I$ and $\Gamma$ link, we naturally refer to the upper and lower unsatisfied vertices relative to their depth in the tree.
\end{definition}

\begin{figure}[ht]
\centering
\begin{subfigure}{0.3\textwidth}
    \centering
    \includegraphics{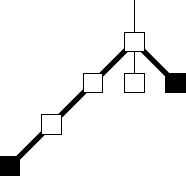}
    \caption{$\Gamma$ link}
    \label{fig:gamma-link}
\end{subfigure}%
\hfill%
\begin{subfigure}{0.3\textwidth}
    \centering
    \includegraphics{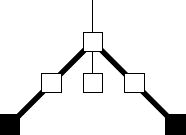}
    \caption{$\Lambda$ link}
    \label{fig:lambda-link}
\end{subfigure}%
\hfill%
\begin{subfigure}{0.35\textwidth}
    \centering
    \includegraphics{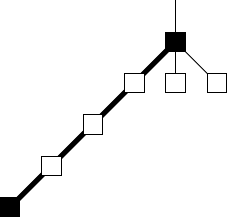}
    \caption{$I$ link}
    \label{fig:iota-link}
\end{subfigure}

\caption{The three different shapes of 4-link are depicted in bold. 
Square vertices represent checks (black for unsatisfied and white for satisfied).
Note that in the following the bold notation will only be used to denote labeled links.}
\label{fig:4links-shapes}
\end{figure}

\begin{definition}
    A \textbf{labeled link} is a link where all edges are labeled. 
    An \textbf{unlabeled link} is a link where at least one edge is unlabeled. 
    From now on we use the notation $\xlink$ to denote a labeled link, and $\link$ to denote an unlabeled link.
\end{definition}

\begin{definition}
An \textbf{alternating chain} is a walk in the tree consisting of a succession of links alternating between labeled and unlabeled (see Figure~\ref{fig:alternatingchain}). Subsequent links can have overlapping edges.

Precisely, an \textbf{alternating chain} is a sequence of links $L_1, L_2, \dots, L_k$ where:\\
(vertices in parenthesis are present only if the incident edge is not dangling)
\begin{align*}
&\forall L_i= (v_0)e_0\dots e_{t-1} v_{t}, \enskip L_{i+1} =  v'_{0} e'_{0}\dots e'_{t'-1}(v'_{t'}) : \qquad v_{t} = v'_{0}\\
&L_i, L_{i+1} \qquad\text{cannot be both labeled or both unlabeled}
\end{align*}

\end{definition}

\begin{figure}[ht]
\center
\includegraphics{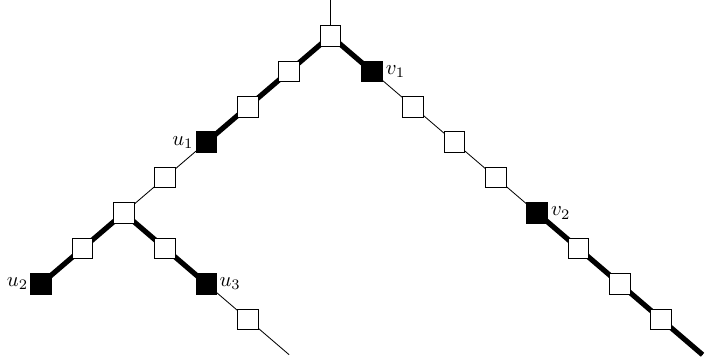}
\caption{An example of an alternating chain $(\_ \link u_3 \overset{\Lambda}{\xlink} u_2 \link u_1 \overset{\Gamma}{\xlink} v_1 \link v_2 \overset{I}{\xlink} \_)$ going from one dangling edge to another (Notice how the two edges above $u_2$ are used twice, in $u_1 \link u_2$ and $u_2 \xlink u_3$). 
Labeled edges are depicted in bold.
All three types of 4-links are depicted here, the $\Lambda$ and $\Gamma$ links are \textbf{proper}, and the labeled $I$-link is \textbf{dangling}.} 
\label{fig:alternatingchain}
\end{figure}
\begin{lemma}[Walk Inversion Lemma] \label{lem:walk-inversion}
Let $C$ be a configuration for a decoding tree $\Tiq$ and syndrome $\synd$, and $W=e_1\dots e_k$ a walk  in $\Tiq$ such that both its endpoints are dangling edges. 
 Let $C':E(\Tiq)\rightarrow \{0,1\}$ be the map  defined by
 $$\forall e \in E(\Tiq),\, C'(e) = (C(e)+n_e) \mod 2,\quad \text{where } n_e=|\{e_i=e, e_i \in W\}|$$ 
 Then  $C'$ is a configuration for $\synd$.  
Put differently, inverting the label of each edge as many times as it appears in the walk preserves the property of being a configuration for a syndrome.
\end{lemma}
\begin{proof}
    Any vertex touches an even number of edges (counting multiplicities) in $W$. 
    Hence inverting the labels as many times as they appear along $W$ will keep the parity of all the vertices in the new configuration.
\end{proof}
\subsubsection{Minimal Configurations in the MS Decoding Tree}

In this section we prove the existence of minimal configurations with special properties, which will be used later in the proof of  Theorems~\ref{thm:limitedscope} and~\ref{thm:4inarow}.

\begin{lemma}[4-links Lemma for root-unlabeled configuration]\label{lem:unlabeled-ge4} 
 Given a syndrome where all unsatisfied checks are at distance 4 or more from one another, there exist a minimal configuration $\Co^{\leq 4}$ where all labeled links are of length $\leq 4$.
\end{lemma}

\begin{proof}[Proof (Sketch).]
We give the proof for the existence of $\Co^{\leq 4}$ in the form of an algorithm that takes any minimal configuration $\Co^*$ and transforms it into a minimal configuration $\Co^{\leq 4}$ that satisfies the condition that all labeled links are of length at most 4. 
The algorithm works as follows: as long as there is a labeled link of length $\ge 5$, we find a path from dangling edge to dangling edge containing it, then use Lemma~\ref{lem:walk-inversion} to remove it. The path is chosen so that the new configuration stays minimal.
The details of the proof can be found in Appendix \ref{app:unlabeled-ge4-fullproof}
\end{proof}

We now prove a similar statement for $\Cx$, but only for decoding trees associated with any of the four qubits neighboring an unsatisfied check\footnote{It can be shown that the statement  does not hold for any other qubit.}. 

\begin{lemma}[4-link Lemma for root-labeled configuration]\label{lem:labeled-ge4}  Given a syndrome where all unsatisfied checks are at distance 4 or more from one another, for a decoding tree $\Tiq$ where $q$ is neighboring an unsatisfied check $c$, there exists a minimal configuration $\Cx^{\leq 4}$ where all labeled links are of length $\leq 4$. 
\end{lemma}

\begin{proof}[Proof (Sketch).]
Starting with a minimal configuration $\Cx^*$, and using the algorithm from the proof of Lemma~\ref{lem:unlabeled-ge4}, one can remove all the length $\ge 5$  labeled links not containing the root. Note that even if the algorithm could be applied on a labeled link containing the root, this would yield an unlabeled configuration which is not what we want.
Supposing that there is a labeled link containing the root of length $\ge 5$, it will also be removed using the walk inversion Lemma, but the path used will be created differently.
The details of the proof can be found in Appendix \ref{app:labeled-ge4-fullproof}. 
\end{proof}

\begin{lemma}\label{lem:c2c}
    Given a syndrome where all unsatisfied checks are all at distance at least 5, and a configuration $C^{\le 4}$ such that all labeled links are of length $\le 4$, then all proper labeled links have endpoints associated with a same check $c$ (\idest, all proper links are of the form  $u_c \xlink v_c$ for some unsatisfied check $c$).
\end{lemma}

\begin{proof}
    Since all labeled links are of length $\leq 4$ and all unsatisfied checks are at distance $\ge 5$, any proper labeled link necessarily joins vertices associated to the same check (by following a path around a plaquette of the toric code). 
\end{proof}

\subsubsection{Proof of Theorem~\ref{thm:limitedscope} (MS Local Blindness)}
\label{sec:proof-thm-MS-local-blindness}

Using configurations  from Lemma~\ref{lem:unlabeled-ge4} and~\ref{lem:labeled-ge4}, and the property of the labeled links from Lemma~\ref{lem:c2c}, we can now prove Theorem~\ref{thm:limitedscope}.

\begin{proof}

Let $\synd$ be a syndrome such
that all the unsatisfied checks are at distance $
\ge 5$ from each others, $c$ be an unsatisfied check in $\synd$, and $q \in \mathcal{N}(c)$. 
Let $\sc$ be the fake syndrome (defined in Section~\ref{subsec:local-blindness}) where only $c$ is unsatisfied , and $\synd^{\neq c} = \synd - \synd^c$ (informally, we can think of $\synd$ as the disjoint union of $\synd^c$ and $\synd^{\neq c}$).

Let also $\Tiq$ be the decoding tree associated to $q$ at iteration $i$. 
On  $\Tiq$, we now consider some minimal configurations $\Cx, \Co$ associated to syndrome $\synd$, and  $\Cxc,\Coc$ associated to  $\sc$, such that all four configurations satisfy Lemmas~\ref{lem:unlabeled-ge4}, \ref{lem:labeled-ge4}, and~\ref{lem:c2c}.
Now, considering $\Cx$ (the  same will hold for $\Co$), we define ${\Cx}_{|c},{\Cx}_{|\neq c}:E(\Tiq)\rightarrow\{0,1\}$  where: 

\begin{itemize}
    \item[(i)] ${\Cx}_{|c}$ is defined by the labeled links in $\Cx$ that have all their vertex endpoints equal to $c$ (\idest, the proper labeled links of the form $u_c \xlink v_c$ and the dangling labeled links of the form $u_c \xlink \bottom$).
    Formally, $${\Cx}_{|c}(e)=1 \iff e \in u_c \xlink v_c \text{ or } e \in u_c \xlink \bottom$$
    Clearly ${\Cx}_{|c}$ is a configuration for $\synd^c$.
    \item[(ii)] ${\Cx}_{|\neq c}$ is defined by the remaining labeled links (that is, labeled links with endpoint(s) different from $c$).
    Alternatively, 
    $${\Cx}_{|\neq c}= \Cx - {\Cx}_{|c}$$
    Using the fact that all proper links are between vertices associated to the same check, it holds that ${\Cx}_{|\neq c}$ is a configuration for $\synd^{\neq c}$.
\end{itemize}

Since $\Cx$ is minimal, we conclude that ${\Cx}_{|c}$ is a minimal root-labeled configuration for $\synd^c$, and $ {\Cx}_{|\neq c}$ is a minimal root-unlabeled configuration for $\synd^{\neq c}$.

Similarly, we can write $\Co$  as $\Co={\Co}_{|c}+ {\Co}_{|\neq c}$, where ${\Co}_{|c}$ is a minimal root-unlabeled configuration for $\synd^c$, and $ {\Co}_{|\neq c}$ is a minimal root-unlabeled configuration for $\synd^{\neq c}$. 
It follows that:
$$
|{\Cx}_{|\neq c}| = |{\Co}_{|\neq c}|,\quad |{\Cx}_{|c}| = |\Cxc|,\text{ and} \quad    |{\Co}_{|c}| = |\Coc|
$$
Therefore,
\begin{align*}
    \app(q,i) &= |\Co|-|\Cx| \\
            &= |{\Co}_{|c}|+|{\Co}_{|\neq c}| - ( |{\Cx}_{|c}|+|{\Cx}_{|\neq c}|) \\
            &= |{\Co}_{|c}| - |{\Cx}_{|c}| \\
            &= |\Coc| - |\Cxc| \\
            &= \app^c(q,i)
\end{align*}
So the blindness property is satisfied for any qubit $q$ neighboring $c$, which concludes the proof of Theorem~\ref{thm:limitedscope}.  
\end{proof}

\begin{corollary}
    Such a syndrome is undecodable by the MS
\end{corollary}

\begin{proof}
    The blindness property for Theorem~\ref{thm:limitedscope} implies that each of the four qubits around an unsatisfied check has the same same \aposteriori  value (thus the same error estimate), at each iteration, as its horizontal/vertical symmetric. Hence the parity of each unsatisfied check is never satisfied, for any number of iterations, and as such, this syndrome is undecodable by the MS.
\end{proof} 

Finally, we note that the blindness property can actually be proven for qubit-to-check messages entering an unsatisfied check, the computation tree of such a message being given by the either  left of right half of the decoding tree in Figure~\ref{fig:toric_dectree} (with root edge $e_q$ becoming a dangling edge). Consequently, it also applies to check-to-qubit messages going out from an unsatisfied check, since such outgoing messages only depend on the incoming ones.

\subsection{Theorem~\ref{thm:4inarow} : MS Non-Degenerate Decoding Radius}

We now go on to prove Theorem~\ref{thm:4inarow}, which we restate here for the reader's convenience.

\weightfourbis*

We first show that there exist weight 4 non-degenerate errors that are undecodable by the MS. Specifically, we use the error consisting of 4 consecutive horizontal qubits in error, denoted by $\epsilon^4$, and $\synd^4$ its weight 2 syndrome. Note that this error can appear anywhere on the code up to translation, and rotations. 
The proof follows the ideas of Theorem~\ref{thm:limitedscope}. Since Lemmas~\ref{lem:unlabeled-ge4} and~\ref{lem:labeled-ge4} were shown for unsatisfied checks at distance $\ge 4$, they automatically apply in the context of Theorem~\ref{thm:4inarow} for syndrome $\synd^4$. We only need to show a (weaker) analog of Lemma~\ref{lem:c2c}.

We then need to show that all weight $\le 3$ non-degenerate errors are decodable by the MS.
This is done using exhaustive search for small codes, and by reduction to localized errors by analysing the number of iterations needed to decode. 

\subsubsection{Preliminary Lemmas}

We now show the analog of Lemma~\ref{lem:c2c}, only this time the result is a little weaker. But first we show a preliminary Lemma used in the proof.

\begin{lemma}\label{lem:lambda} 
For the syndrome $\synd^4$, and for any decoding tree $\Tiq$ associated to qubit $q$ incident to an unsatisfied check $c$:

Let $u \link v$ be a $\Lambda$-link in $\Tiq$ and $S$ be the smallest subtree of $\Tiq$ containing it, then for any vertex $w \in V(\Tiq) \setminus V(S)$ associated to an unsatisfied check, $$d(u,w)=d(v,w)\ge 6$$
\end{lemma}

\begin{proof}
Suppose there is a $\Lambda$-link $u_c \link v_{c'}$. We now consider the corresponding walk~without~return between the checks $c$ and $c'$ on the toric code. 
There are two possibilities, either $c = c'$ (Figure~\ref{fig:lambda-a}), or $c\neq c'$ (Figure~\ref{fig:lambda-b}).
In any case, let $g$ be the toric code check associated with the grandfather of $u_c$ and $v_{c'}$ in $\Tiq$. On the toric code, let $p$ and $p'$ the qubits incident to $g$ on the related walks without return to $c$ and $c'$. Then any wwr from $g$ to $c$ or $c'$ and starting with $p'' \notin \{p,p'\}$ is of length $\ge 4$. Hence, any vertex associated to an unsatisfied check in $V(\Tiq)\setminus S$ must be at distance $\ge 6$ from $u_c$ and $v_c'$. 
\end{proof}

\begin{figure}[t]
\centering
\begin{subfigure}{0.48\textwidth}
\centering
    \includegraphics{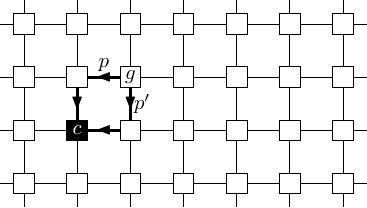}
    \caption{$\Lambda$-link between two vertices associated to the same check, visualized on the toric code.}
    \label{fig:lambda-a}
\end{subfigure}\hfill% 
\begin{subfigure}{0.49\textwidth}
\centering
    \includegraphics{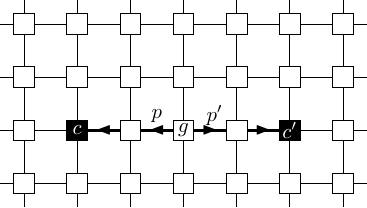}
    \caption{$\Lambda$-link between two vertices associated to different checks, visualized on the toric code.}
    \label{fig:lambda-b}
\end{subfigure}

\caption{The cases of $\Lambda$-link}
\label{fig:4links}
\end{figure}

The following corollary is a straight application of Lemma~\ref{lem:lambda}.

\begin{corollary}\label{corollary:lambda}
   Under the hypothesis of Lemma~\ref{lem:lambda}, let $u\link v$ be an $I$-link with $v$ being its the bottom endpoint. Then $v$ cannot belong to a $\Lambda$-link.
\end{corollary}

The following Lemma is a (weaker) analog of Lemma~\ref{lem:c2c} for the syndrome $\synd^4$.

\begin{lemma}\label{lem:c2c-e4}
For the syndrome $\synd^4$, and for any decoding tree $\Tiq$ associated to qubit $q$ neighboring an unsatisfied check $c$, there exists a minimal configuration such that the endpoints of any $I$ or $\Gamma$ labeled link are vertices associated to the same check.
\end{lemma}

\begin{proof}
Take a configuration having all labeled links of length $\le 4$.
Suppose there is a labeled link between $u^0_c$ and $v^0_{c'}$, and take one deepest if there are several. 
If applying, always consider $u^0$ always be on a higher level than $v^0$.\smallskip
\begin{itemize}
    \item[(i)] If this link is an $I$-link or a $\Gamma$-link, then there must exist an unlabeled $I$-link  $u^0_c \link u^1_c$ where $u^1_c$ is at the same depth as $v^0$. 
Then set $\mathcal{A} = \{u^1_c \link u^0_c, u^0_c \link v^0_{c'} \}$
The link $u^0_c \link u^1_c$ must be unlabeled as otherwise the link $u^1_c \xlink v^0_c$ of length $\ge 5$ would contradict the assumptions of Lemma~\ref{lem:unlabeled-ge4}.

    \item[(ii)] If the link contains the root edge $q$ (in case of a labeled configuration), then there must also exists an unlabeled link  $u^0_c \link u^1_c$ where $u^1_c$ is at the same depth as $v^0$ using the same argument. 
Then set $\mathcal{A} = \{u^1_c \link u^0_c, u^0_c \link v^0_{c'}\}$.
\end{itemize}

Now complete this into an alternating chain from dangling edge to dangling edge by doing the following iterative algorithm that satisfies the property:
\begin{itemize}
    \item[$(P)$] \textit{At each iteration of the algorithm, the chain endpoints are at the same depth on both sides}
\end{itemize}

At a given iteration, suppose without loss of generality that $u^i$ came from a labeled link and $v^i$ from an unlabeled link.
Then pick one labeled link from $v^i$ going down. There must be one and it must be unique since we assumed this is a valid labeling, and by Lemma~\ref{lem:lambda}, this cannot be a labeled $\Lambda$-link.
Also pick an unlabeled link of the same type in the neighborhood of $u^0$. There must exist one, otherwise this would give a labeled link of length $ \ge 5$.\smallskip

This algorithm gives a path from dangling edge to dangling edge.
Now inverting this path will remove the link between $u^0_c$ and $v^0_{c'}$ while maintaining exactly the same weight, hence keeping the minimality of the configuration.
\end{proof}

\subsubsection{Sketch of Proof of Theorem~\ref{thm:4inarow} }

To prove the blindness of the qubits around the unsatisfied checks of $\synd^4$, we cannot use the fact that all proper labeled links are between copies of the same checks. Fortunately, the only time there can be a proper labeled link $u_c \xlink v_{c'}, \ c\neq c'$, it must be a $\Lambda$-link. Using a notion of pruning of the decoding tree, we are able to split the decoding tree into a pruned tree where the only remaining unsatisfied vertices are associated to $c$, and removed subtrees where there might still be vertices associated to $c$, but with the guarantee (that will be the bulk of the proof) that there are no labeled edges between the removed subtrees and the pruned tree. The details of the proof can be found in  Appendix~\ref{app:weightfour_a}.

To show that all non-degenerate errors of weight $\le 3$ are decodable, we do it by a formal argument for codes of larger distance ($d\ge 18$) by introducing the notion of $\delta$-local syndrome, and we verify it numerically for smaller codes ($9\le d < 18$).
We defer the proof to Appendix~\ref{app:weightfour_b}.

\subsection{Conjectures about Normalized MS and BP}

For Normalized MS (NMS), we conjecture that Theorems~\ref{thm:limitedscope} holds for any normalization factor in $]0,1[$.
To support our claim, we did some numerical testing for codes with minimum distance $d \in \{11,24\}$ (giving a ``small'' and a ``larger'' code, odd and even size), normalization factor  $\lambda$ between 0.0625 and 1, with a step of 0.0625, and for a maximum number of decoding iterations equal to 100. Note that since the normalization factors are sums of powers of two, we do not have to worry about numerical instability and check equality up to floating point precision. 
Verifying the statement for all error syndromes with all the unsatisfied checks  at distance  $\ge 5$ even for reasonably small code sizes is not  feasible, so we only tested syndromes of weight 2, satisfying the condition that the two unsatisfied checks are at distance $\ge 5$. 
Also to support even more our claim, we tested the syndrome defined by $\synd(c_{i,j})=1 \iff j = 5i \mod 13$ on toric codes of distance $d=13$ (Figure~\ref{fig:nms_conjecture}) and $d=26$. This is a syndrome ``packed'' with as many unsatisfied checks as possible, all at distance 5 or 6 from each other. In all our numerical experiments, the blindness property is conserved.

\begin{figure}[!thb]
\centering
\includegraphics{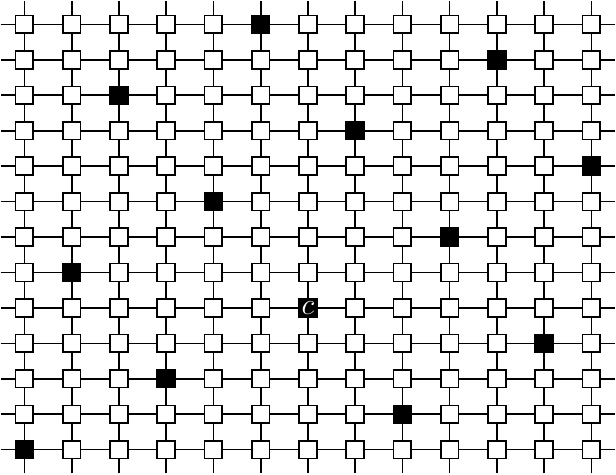}
\caption{Syndrome defined by $\synd(c_{i,j})=1 \iff j = 5i \mod 13$ on a toric code of distance $d=13$. Note that this is not a valid error syndrome (as it is of odd weight), but  it is only used to verify the blindness property. For $d=26$ we get a valid error syndrome.
}
\label{fig:nms_conjecture}
\end{figure}

For BP, we know the blindness statement of Theorem~$\ref{thm:limitedscope}$ to be false, since  BP takes into account all the configurations in the decoding tree (for a given syndrome), and not only a minimal one (see~\cite{Wiberg96} for more details). However, the minimal configurations are still somehow dominant in the \aposteriori reliabilities computed by the BP, so we wonder if a similar statement could be made, perhaps using a relaxed notion of blindness. So far, our preliminary numerical results suggest that 
syndromes satisfying the hypothesis of Theorem~$\ref{thm:limitedscope}$ are not decodable by BP, and moreover $\displaystyle\lim_{i\rightarrow\infty}\left(\app(q,i) - \app^c(q,i)\right) = 0$, which is likely explained by the fact that minimal configurations become more dominant as the decoding tree size increases.

\section{Stabiliser-Blowup Pre-Processing }

In Section \ref{sec:limited-scope}, we made a point that there is more than just degeneracy preventing the convergence of the MS, and that there is no hope to correct adversarial errors of weight 4 or more on the toric code using MS. 
While all the non-degenerate errors of weight $\leq 3$ are decodable by the MS (see Appendix \ref{app:weightfour_b}), numerical simulations show that most degenerate errors are not (see Appendix \ref{app:add-res-degen} for a detailed discussion). 
Here we present a new pre-processing that runs in linear time and allows to correct all errors of weight up to 3.
The idea of the pre-processing is to infer where the decoder might fail by looking at local degenerate syndromes. We then apply a code modification, corresponding to a change of variables, that locally removes the degeneracy to allow the decoder to converge.

\subsection{The Stabiliser-Blowup Graph Modification}

\begin{figure}[t]
\centering
\begin{subfigure}{.46\linewidth}
\centering
    \includegraphics[scale=0.95]{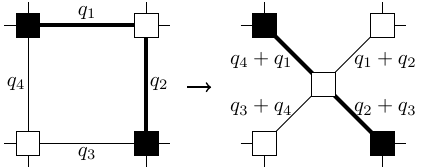}
    
    \vspace*{-5mm}
    \caption{}
\end{subfigure}\hfill%
\begin{subfigure}{.46\linewidth}
\centering
    \includegraphics[scale=0.95]{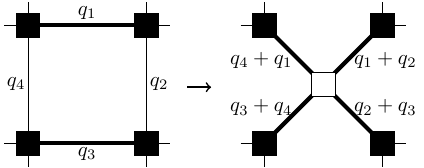}

    \vspace*{-5mm}
    \caption{}
\end{subfigure}
\caption{Stabiliser-blowup for degenerate errors of weight 2: 
the four edges around the original plaquette are replaced by four new edges, connecting the checks around the plaquette to a new check in the middle of the plaquette. Labels $q_i$'s on the edges of the original plaquette are interpreted here as binary random variables, equal to $1$ if the corresponding qubit is in error, and $0$ otherwise. Each new edge is labeled by the sum of the $q_i$'s incident to the same check. 
(a) The diagonal weight-2 syndrome, shown on the original plaquette together with one of its possible explanations ($q_1 = q_2=1$) on the left. After the blowup, the syndrome is uniquely explained as $q_4+q_1 = q_2+q3 = 1$. 
(b) The square weight-4 syndrome, shown on the original plaquette together with one of its possible explanations ($q_1 = q_3=1$). After the blowup, the syndrome is uniquely explained as $q_4+q_1 = q_1+q_2 = q_2+q3 = q_3+q_4 = 1$. } 
\label{fig:sb-examples}
\end{figure}

The stabiliser-blowup is a decoding graph modification that allows to locally remove the degeneracy of the code around a plaquette. 
In the following, we use it as a preprocessing to remove the degeneracy of degenerate weight 2 errors (depicted in Figure~\ref{fig:sb-examples}) and degenerate weight 3 errors (depicted in Figure~\ref{fig:sb-example-c}).

\begin{figure}[t]
\center
\includegraphics[scale=0.95]{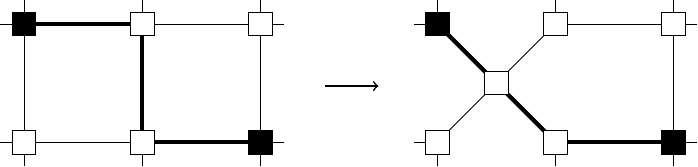}
\caption{Stabiliser-blowup for degenerate errors of weight 3: an example of a  weight-3 degenerate error on the original toric code (generating a weight-2 syndrome) becoming a weight-3 non-degenerate error after the blowup. }
\label{fig:sb-example-c}
\end{figure}

We  remove the four edges  associated to the qubits around a given plaquette and replace them by four new edges connecting the four corner checks of the plaquette to an additional (satisfied) check in the middle. Each new edge corresponds to the parity of the two removed edges around a given check. Thus, the sum of the four new edges must be even, hence the addition of the satisfied middle check connecting them.
We later show that any solution satisfying the syndrome on the modified graph will easily yield a valid solution for the original graph, and this solution is unique up to stabiliser equivalence.
As a remark, this code change should not be applied on two neighboring plaquettes, as this would create new degeneracy issues between them (see Figure~\ref{fig:sb-twice}).
Therefore, we use the stabiliser-blowup as a pre-processing, and only apply the blowup where the error syndrome suggests that degeneracy issues might occur. 
One could ask why those local degeneracy issues are not directly solved ``by hand'' as a pre-processing,
by assigning a fixed value to one of the qubits thus removing the degeneracy. But this could lead to a degradation of the decoding by systematically increasing the weight of the outputted error, thus increasing the likelihood of logical errors, so we instead rely on the MS to decide on the error pattern. 

\begin{figure}[t]
\center
\includegraphics[scale=0.95]{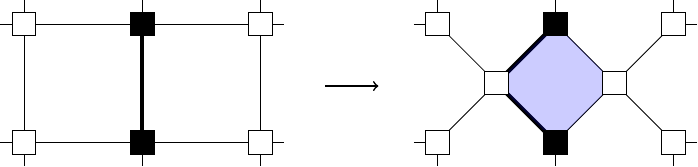}
\caption{Stabiliser-blowup on adjacent plaquettes leads to undecodable  weight-1 errors: the weight-1 error on the original toric code becomes a weight-2  degenerate error after the blowup.}
\label{fig:sb-twice}
\end{figure}
\subsection{The Algorithm}
In order to correct all errors of weight $\leq 3$, it is enough to consider the syndromes of \emph{local} degenerate errors of weight 2 and 3. 
Whenever one of those patterns appears around a plaquette, we do the blowup, except if an adjacent plaquette (sharing an edge) is already blown-up.

To remove all errors of weight $\leq 3$, we use the following \textbf{greedy heuristic}:
We go through all plaquettes  in the lexicographical order of their coordinates. We do this three times, corresponding to steps (i) (ii) (iii) below, each time checking for a different set of patterns. 
If  one of the patterns looked for at this step is found in the neighborhood of a plaquette $p$, we apply the blowup on $p$ if it has no adjacent plaquette that has been blown-up.
The local neighborhood of each plaquette is labeled in Figure~\ref{fig:diagonal-fail}.

\begin{figure}[ht]
    \centering
    \includegraphics{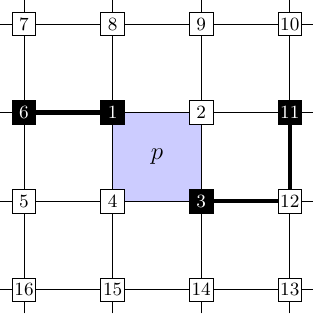}
    \caption{The neighborhood of plaquette  $p$, corresponding to the 16 checks in its surrounding. 
    In bold the corner case of the heuristic in (i) is depicted.}
    \label{fig:diagonal-fail}
\end{figure}

Each check is given a number, and in the following we associate to each check a true/false value depending on whether it is unsatisfied/satisfied.
\begin{itemize}
    \item[(i)] Weight 2 degenerate errors with condition on diagonal
    \begin{itemize}
        \item $\synd(c_1) \wedge \synd(c_2) \wedge \synd(c_3) \wedge \synd(c_4)$
        \item $\overline{\synd(c_1)} \wedge \synd(c_2) \wedge \overline{\synd(c_3)} \wedge \synd(c_4) \wedge \Big( \synd(c_9)\oplus \synd(c_{11}) \oplus \synd(c_5) \oplus \synd(c_{15}) = 0 \Big)$
        \item $\synd(c_1) \wedge \overline{\synd(c_2)} \wedge \synd(c_3) \wedge \overline{\synd(c_4)} \wedge \Big(\synd(c_6)\oplus \synd(c_8) \oplus \synd(c_{12}) \oplus \synd(c_{14}) = 0 \Big)$
    \end{itemize}
    
    \item[(ii)] Weight 2 diagonal degenerate errors without condition
\begin{itemize}
    \item $\overline{\synd(c_1)} \wedge \synd(c_2) \wedge \overline{\synd(c_3)} \wedge \synd(c_4)$
    \item $\synd(c_1) \wedge \overline{\synd(c_2)} \wedge \synd(c_3) \wedge \overline{\synd(c_4)}$
\end{itemize}

    \item[(iii)] Weight 3 degenerate errors - patterns not yet catched by (i) and (ii)
\begin{itemize}
    \item $\synd(c_1) \wedge \overline{ \synd(c_2)} \wedge \overline{ \synd(c_{11})} \wedge \overline{\synd(c_4)} \wedge \overline{\synd(c_3)} \wedge \synd(c_{12})$ 
     \item $\overline{\synd(c_1)} \wedge \overline{\synd(c_2)} \wedge \synd(c_{11}) \wedge \synd(c_4) \wedge \overline{\synd(c_3)} \wedge \overline{\synd(c_{12})}$ 
      \item $\synd(c_1) \wedge \overline{\synd(c_2)} \wedge \overline{\synd(c_4)} \wedge \overline{\synd(c_3)} \wedge \overline{\synd(c_{15})} \wedge \synd(c_{14})$ 
      \item $\overline{\synd(c_1)} \wedge \synd(c_2) \wedge \overline{\synd(c_4)} \wedge \overline{\synd(c_3)} \wedge \synd(c_{15}) \wedge \overline{\synd(c_{14})}$ 
\end{itemize}
\end{itemize}

The diagonal condition of (i) comes from a corner case that would make the heuristic fail if not property taken into account.  
This is illustrated in Figure~\ref{fig:diagonal-fail}, where, if the plaquette $p$ is blown-up, then the plaquette on the right could not be blown-up in its turn, and the MS would fail to converge.
Clearly (i) and (ii) are enough to take care of all weight 2 degenerate errors. There are some weight 3 errors which are also taken into account by (i) and (ii), the remaining ones being taken care of by (iii).

\SBMS*

The proof of the Theorem is given in Appendix~\ref{app:sbms-proof}, by using the tools introduced to prove the second part of Theorem~\ref{thm:4inarow}.

\subsubsection{Numerical Results}

Here we plot the error correction performance of the MS versus the SB+MS. One can see that the slopes match the expected values. Indeed, the MS is able to converge on all errors of weight 1 but fails on some errors of weight 2. The dominant term is thus $p^2$ and the slope matches it for small $p$. For the SB+MS since the smallest undecodable error is of weight 4, and the slope matches the $p^4$ curve.

\begin{figure}
    \centering
    \includegraphics[scale=0.8]{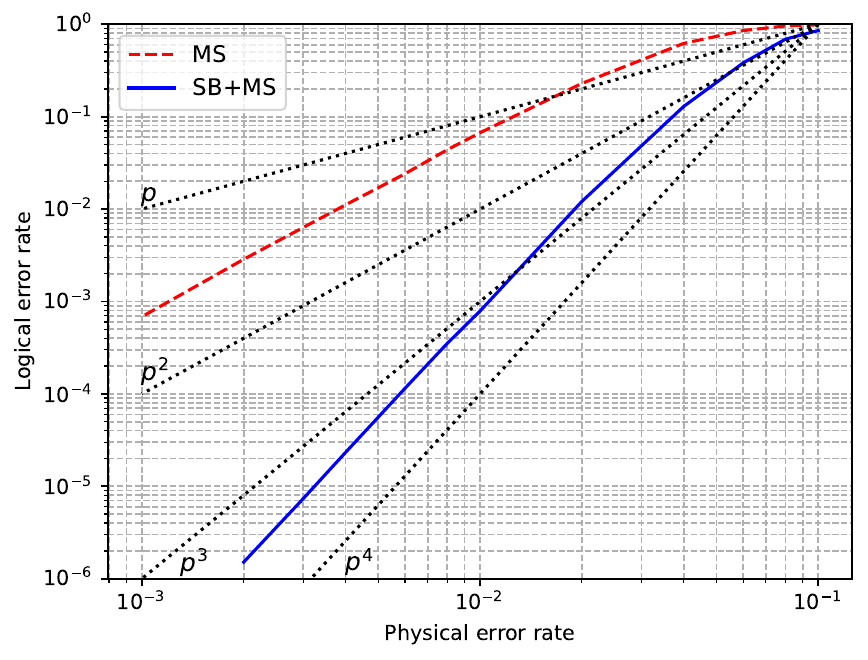}
    \caption{Numerical results for  MS and  SB+MS decoding on the toric code of distance~11. A maximum number of 15 decoding iterations are used for the MS decoding. The slope of the MS is $p^2$ since the smallest undecodable errors are of weight 2, and the slope of the SB+MS is $p^4$ since the smallest undecodable errors are of weight 4.}
    \label{fig:blowup-numres}
\end{figure}

\subsubsection{SB+MS+OSD}

One of the most used post-processing is Ordered Statistic Decoding (OSD), proposed in \cite{panteleev2021degenerate}.
Numerically, MS+OSD  is a good approximation of the classical ML for the toric code (see for example \cite{roffe2020decoding}). However since the OSD post-processing is very costly, namely $\mathcal{O}(n^3)$,\footnote{We acknowledge that the theoretical complexity of Gaussian elimination is $\mathcal{O}(n^{2.38})$ (see for example \cite{golub2013matrix}), but such algorithms are not practical as they suffer from huge constant costs, and for real life use cases, the Gaussian elimination algorithm is used.} it is interesting to run the best possible decoder beforehand to avoid unnecessary use of OSD. Looking at the curves from Figure~\ref{fig:blowup-numres}, this yields a quadratic improvement in the number of calls of OSD, only needing to call OSD with probability $\approx p^4$ compared to $\approx p^2$.  Since our pre-processing is linear, this implies a significant improvement on the average decoding complexity for small $p$, going from $\mathcal{O}(p^2 n^3 + n\log(n))$ to $\mathcal{O}(p^4 n^3 + n\log(n) )$.

\section{Conclusion}

Using the powerful framework of decoding tree, we were able to show that there is an inherent limitation to the error-correction capability of the MS on the toric code.
This was done by introducing the  notion of blindness, which not only emphasize the fact that the decoder is unable to converge, but also prove that the spreading of the messages in the code is inherently limited.
We also showed that by using a low cost syndrome-aware pre-processing we were able to make the decoding radius and the non-degenerate decoding radius of the MS match.
In the vein of what had been done recently on the limitations of the renormalization decoder on the toric code in~\cite{rozendaal2024analysis}, to the best of our knowledge this is the first formal result explaining why the MS decoder does not perform well on the toric code, and we also conjecture that similar results could be obtained for other iterative decodings (NMS and BP) that have similar decoding processes, and can also be explained using the decoding tree, although in a different manner that do not allow the proofs to transfer directly. 

The blindness property we proved here comes from a strong imbalance between the girth and the diameter of the graph. This can be seen as a \emph{graph degeneracy}, complementing the standard notion of \emph{code degeneracy}  (\idest, classical minimum distance smaller that the quantum minimum distance), the latter being independent of the Tanner graph representation of the code. It is worth noticing that for well-constructed classical LDPC codes, \exgrata, constructed by the progressive edge growth algorithm~\cite{hu2005regular}, such a graph degeneracy does not happen: the length of the shortest cycle passing through a given check is greater than or equal to the maximum distance between that check and all the others. This ensures fast spreading of messages in the Tanner graph.  
We suspect that the graph degeneracy could have blindness effects for more general quantum LDPC codes, although the imbalance between girth and diameter would be reduced for finite-length Tanner graphs with increased connectivity degree.

Additionally, it might be interesting to explore the stabiliser-blowup pre-processing on its own, to see if it can be applied to a wider class of codes. There should be a straightforward generalisation for all planar codes, but we do believe that the pre-processing could be extended to generic quantum LDPC codes.

\section*{Acknowledgment}

The authors would like to thank Leonid Pryadko and Xingrui Liu for usefull discussions during the elaboration of this work.
This work was supported by the QuantERA grant EQUIP (French ANR-22-QUA2-0005-01), and by the Plan France 2030 (NISQ2LSQ project, ANR-22-PETQ-0006).

\newpage
\appendix

\section{Additional Results on Small Weight Non-Degenerate and Degenerate Errors}\label{app:add-res-degen}

\subsection{Weight 2 Degenerate Errors are Undecodable by MS}

There are only two degenerate errors of weight 2 up to translation and rotations (depicted in Figure~\ref{fig:weight2-sym}). To show that they are both undecodable by the MS, we use the fact that both syndromes are invariant under some symmetries.
The two cases are depicted in Figure~\ref{fig:weight2-sym} (a) and (b).
For the weight 2 syndrome of (a), it is clear that by the symmetry about the $\ell$ line,  $\app(q_k,i) = \app(q'_k,i)$ for $k\in \{0,1\}$,  $\forall i \geq 0$. 
Hence the estimate error around $c$ will never satisfy the syndrome $\synd$. 
A similar argument applies for (b).

\begin{figure}[ht]
    \centering
    \includegraphics{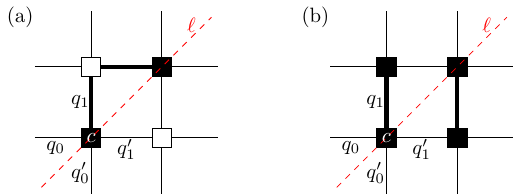}
    \caption{Undecodable degenerate errors of weight 2}
    \label{fig:weight2-sym}
\end{figure}

\subsection{A Weight 3 Degenerate Error Decodable  by MS}

Although we said earlier that degenerate syndromes are key to understand the failures of the MS, and we made it clear using numerical simulations in the above subsection, it is not true that degenerate implies undecodable by MS. In Figure~\ref{fig:decodable-degenerate} we provide the smallest example of a degenerate syndrome that is decodable (in 1 iteration) by the MS.

\begin{figure}[ht]
\center
\includegraphics{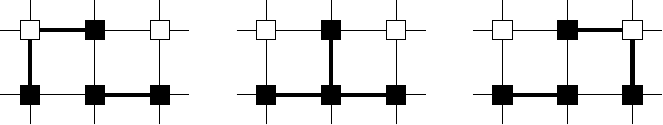}
\caption{Equivalent weight 3 decodable degenerate errors (generating the same degenerate syndrome of weight 4):  Although the three errors are all of weight~3, the MS will converge to the middle error in just one iteration (this can be easily checked, by computing ``by hand'' the \aposteriori  value of each qubit after one iteration). 
}
\label{fig:decodable-degenerate}
\end{figure}

\section{Proof of Lemmas~\ref{lem:unlabeled-ge4} and~\ref{lem:labeled-ge4}}\label{app:proofs}

We first start by showing a few technical lemmas before proving  Lemmas~\ref{lem:unlabeled-ge4} and~\ref{lem:labeled-ge4}.

\subsection{Technical Lemmas}

\begin{definition}
   A labeled link is said \textbf{maximal}  if it cannot be extended into a longer labeled link.
   It is said \textbf{robustly maximal}  if there is no other labeled link sharing at least one edge with it that is of greater length.
\end{definition}

\begin{lemma}\label{lem:maximal-labeled}
  Given a configuration $C$,   and a maximal proper labeled link $u\xlink v$. 
  Then at least 2 of the 3 edges going down from $v$ are unlabeled.
\end{lemma}

\begin{proof}
     Since $v$ is the endpoint of a maximal labeled link, it is either true that:
    \begin{itemize}
        \item[(i)] $u$ belongs to the subtree rooted in $v$, in which case the only labeled edge going down from $v$ is the one belonging to the labeled link $u\xlink v$, as otherwise $u \xlink v$ could be extended and would not be maximal.
        \item[(ii)] On the other hand, if $u$ does not belong to the subtree rooted in $v$, all three edges below $v$ are unlabeled as otherwise the labeled link $u\xlink v$ could be extended and would not be  maximal.
    \end{itemize}
\end{proof}

\begin{lemma}\label{lem:exists-I}
Given a configuration $C$ for a syndrome where all unsatisfied checks are at distance $\ge 4$, and a maximal proper labeled link $u\xlink v_c$, where $v_c$ is at distance $\ge 5$ from the bottom. 
Then there exists an unlabeled $I$-link $v_c \link v'_c$ such that $v_c$ is its upper endpoint.
\end{lemma}

\begin{proof}
    By Lemma~\ref{lem:maximal-labeled}, there must be one unlabeled edge $\edgeq$ going down from $v_c$.
    In the toric code, there exists a path of length 4 that starts from $c$, goes through qubit $q$ and loops back to $c$. Hence, since $\edgeq$ is unlabeled, there exist an unlabeled $I$-link $v_c \link v'_c$. 
\end{proof}

\begin{lemma}\label{lem:exists-unlabeled-dangling}
    Given a minimal configuration $C$ for a syndrome where all unsatisfied checks are at distance $\ge 4$, and a maximal proper labeled link $u \xlink v$ of length $\ge 4$, where $v$ is a distance $d\leq 4$ from the bottom.
    Then any dangling link $v \link \bottom$ is of length $d$ and unlabeled.
\end{lemma}

\begin{proof}
We note that the subtree rooted in $v$ does not contain any unsatisfied vertex (since all vertices in the subtree are at distance $\leq$ 3 from $v$). In particular, the proper link $u \xlink v$ must come from above. Since $u \xlink v$ is  maximal,  it cannot be extended inside the subtree rooted in $v$, and since $C$ is minimal it follows that all the edges in this subtree are unlabeled. Thus, any dangling link $v\link \bottom $ is of length $d$ and unlabeled. 
\end{proof}

\subsection{Proof of Lemma~\ref{lem:unlabeled-ge4}}\label{app:unlabeled-ge4-fullproof}

\begin{proof}

Let $\Co^*$ be any minimal root-unlabeled configuration on the decoding tree $\Tiq$. 

\paragraph{Case 1 (Proper labeled link $u^0 \xlink v^0$):} 
Suppose there is a proper labeled link $u^0 \xlink v^0$ of length $\ge 5$.  
If there are several choices,  it is always possible to take one that is robustly maximal and as deep as possible (\idest, that in the subtrees rooted in $u^0$ and $v^0$, there are no labeled links of length $\ge 5$, except possibly $u^0\xlink v^0$).
 
Now recursively create an alternating chain $\mathcal{A}$ containing the labeled link $u^0 \xlink v^0$ that goes from one dangling edge to another. This alternating chain will have the property: 
$$ P:  \text{all labeled links in $\mathcal{A}$ are maximal}$$
We start with $\mathcal{A} = \{ u^0 \xlink v^0 \}$. Set $i=0$ and do recursively starting from $u^0$ (and later do the same for $v^0$):

\begin{itemize}
    \item[(i)] 
    $i$ even: an unlabeled link must be appended. 
If $u^i$ is at distance $\ge 5$ from the bottom, then there exists an unlabeled $I$ link $u^i \link u^{i+1}$ such that $u^i$ is its upper endpoint (by Lemma~\ref{lem:exists-I}). Set $\mathcal{A} = \mathcal{A}\cup \{ u^i \link u^{i+1} \}$\footnote{Strictly speaking, an alternating chain is an ordered sequence of links so the union should be seen as preppending/appending the link to the ordered list.}.

Otherwise, $u^i$ is at distance $d \leq 4$ from the bottom, and there exists a dangling unlabeled link $u^i\link \_$ of length $d$ (by Lemma~\ref{lem:exists-unlabeled-dangling}). Set $\mathcal{A} = \mathcal{A}\cup \{u^i \link \bottom \}$ and terminate the algorithm on this side.
    \item[(ii)]
    $i$ odd: a labeled link must be appended.
Take any labeled link having $u^i$ as endpoint (we show below that is necessarily maximal).  If it is a proper labeled link $u^i\xlink u^{i+1}$, set $\mathcal{A} = \mathcal{A}\cup \{u^i \xlink u^{i+1}\}$.

Otherwise, it is a dangling link $u^i\xlink \_$. Set $\mathcal{A} = \mathcal{A}\cup \{u^i \xlink \_ \}$  and terminate the algorithm on this side.
\end{itemize}

Since $u^0 \xlink v^0$ was taken to be as deep as possible, any time a labeled link is appended, it must be of length $\leq 4$, and if it is a proper link, since all vertices associated to unsatisfied checks are at distance $\ge 4$, it must be of length exactly 4, and thus maximal, so property $P$ is conserved throughout the algorithm.\smallskip

This algorithm will always terminate because, one of the following is always true:\smallskip
\begin{itemize}
\item[(i)] When adding an unlabeled link, either the algorithm terminates on this side or an $I$-link was appended, in which case the depth increases by 4.

\item[(ii)] When adding a labeled link,  either the algorithm terminates on this side, or a proper labeled link $u^i \xlink u^{i+1}$ of length 4 has been appended, in which case the depth decreases by at most 3 as $u^i$ is already the bottom endpoint of an unlabeled $I$ link so it cannot also be the bottom endpoint of a labeled $I$ link. 
\end{itemize}
Thus the algorithm must end.

We now have an algorithm to construct an alternating chain that goes from dangling edge to dangling edge. This is not enough to use the inversion Lemma directly, so we further modify it to ensure that there is no dangling unlabeled link of length $\ge3$.

Suppose there is one unlabeled dangling link of length $\ge 3$, without loss of generality write it $u^k \link \_$. 
Then since $u^k$ it at distance $\ge 3$ from the bottom, $u^k$ must be the upper endpoint of a $\Gamma$-link $u^k \link u^{k+1}$. This $\Gamma$-link must be unlabeled because :
\begin{itemize}
    \item[(i)] If $k=0$, $u^{k-1}=v^0$ and $u^{k-1}\xlink u^{k+1}$ would be longer than $u^0\xlink v^0$ and this would give a contradiction to the fact that $u^0\xlink v^0$ is robustly maximal.
    \item[(ii)] If $k\neq 0$, then $u^{k-1}\xlink u^{k+1}$ would be a labeled link of length $\ge 6$ which would contradict the fact that $u^0\xlink v^0$ was taken as deep as possible.
\end{itemize}
 
Now set $\mathcal{A} = \Big(\mathcal{A}\setminus \{ u^k\link \_ \} \Big) \cup \{u^k \link u^{k+1}\}$.

At this point, we arrive at unsatisfied check $u^{k+1}$ at distance $\leq 2$ of the bottom, and we need to add a labeled dangling link to finish our alternating chain. 
Since $u^{k+1}$ is unsatisfied, it is must be the endpoint of a labeled link.
It is impossible that $u^{k+1}$ is the bottom endpoint of a labeled $I$ or $\Gamma$ link as it would contradict the fact that $u^k \link u^{k+1}$ is unlabeled. So it must either be a labeled dangling link, in which case we are done, or it is a labeled $\Lambda$-link. In that case we append it and then finish the alternating chain by appending a length 2 unlabeled dangling link that must exist according to Lemma~\ref{lem:exists-unlabeled-dangling}.
We now have the alternating chain with all the required properties to finish the proof.

We look at the walk given by the alternating chain (with possible back-and-forth). We do the path inversion according to Lemma~\ref{lem:walk-inversion}.
This gets rid of the labeled link $u^0 \xlink v^0$ of length $\ge 5$.

To conclude the proof, we now argue that the configuration obtained by applying the walk inversion Lemma is minimal.

The added cost of labeling each previously unlabeled proper 4-link will be compensated by unlabeling the next proper labeled 4-link, if it exists. As they are both of length 4 and there will be as many edges being unlabeled as being relabeled\footnote{It should be easy to convince oneself that this is also true when there are overlapping edges.}.
We only need to deal with the dangling links at both ends of the alternating chain.
We show that whatever weight will be added at both end is compensated by the unlabeling of $u^0 \xlink v^0$, thus keeping the minimality of the configuration.

There are several cases to consider. On both sides, the chain can end with either an unlabeled link of length 1 or 2, or a labeled link of length 1, 2, 3 or 4, preceeded by an unlabeled $I$ link of length 4. Consider one end of the alternating chain, and denote these cases as $U_1,U_2$ and $L_1,\dots,L_4$ respectively.

\begin{itemize}
\item $U_1$,$U_2$: adds 1 or 2 to the weight
\item $L_1$: adds 3 to the weight.
\item $L_2$,$L_3$,$L_4$: adds respectively 2, 1 or 0 to the weight
\end{itemize}

The only troublesome case is $L_1$, since if that where to happen on both ends of the alternating chain, the added weight of 6 might not be compensated by unlabeling of a weight 5 labeled link $u^0 \xlink v^0$. A simple parity argument takes care of this corner case: in case the $u^0 \xlink v^0$ is of length 5, $u^0$ and $v^0$ must be at a depth of different parity in the tree\footnote{The only case where this is not true is when the link contains the root but this case is purposefully avoided in the Theorem}. Since at each step adding a 4-link conserves the parity of the depth, it cannot be that the case $L_1$ happens on both ends. 

\paragraph{Case 2 (Dangling labeled link $u^0 \xlink \_$):} Suppose there is a dangling labeled link with endpoint $u^0$ of length $\ge 5$. The same algorithm can be applied from $u^0$, except the argument for the minimality has to be (easily) adapted from Case 1.

\paragraph{Recursive Application of the Algorithm:}
By applying the algorithm as long as there are labeled links of length $\ge 5$ (except possibly if it contains the root), this algorithm transforms any minimal root-unlabeled configuration in a minimal root-unlabeled configuration where all labeled links not containing the root are of length $\leq 4$.
\end{proof}

\subsection{Proof of Lemma~\ref{lem:labeled-ge4}}\label{app:labeled-ge4-fullproof}

\begin{proof}
Starting with a minimal configuration $\Cx^*$, and using the algorithm from the proof of Lemma~\ref{lem:unlabeled-ge4}, one can remove all the length $\ge 5$ labeled links not containing the root. 

Now suppose there is a labeled link of length $\ge 5$ containing the root $e_q$.
Since $q$ is neighboring an unsatisfied check $c$, we know that there is an unsatisfied vertex $u_c$ just next to $\edgeq$.

Note that the edges going down from $u_c$ must be unlabeled. 
This is because, since $\edgeq$ is labeled, to keep the parity of $u_c$, it must be that two edges below $u_c$ are labeled, hence, that there are two labeled links $u_c \xlink u'$ and $u_c \xlink u''$ below it. But then the labeled link $u' \xlink u''$ does not contains the root and must be of length $\ge 8$ which is a contradiction. So we conclude that 
the longest labeled link starting from $u_c$ and going through the root $\edgeq$ is robustly maximal.

\paragraph{Case 1 (proper labeled link $u_c \xlink v$):}
      Suppose the labeled link containing the root is of length $\ge 5$ of the form $u_c \xlink v$. 
      Let $u_c$ and $w$ be the vertices incident to $\edgeq$. Since all the labeled links not containing the root are of length $\leq 4$, 
      it follows that from the three edges going down from $w$, only the one contained in the path going to $v$ is labeled.  Then, we can take an  unlabeled edge $e$ going down from $w$, such that there exists a proper unlabeled 4-link $u_c \link v'_c$ containing $\edgeq$ and $e$. 

     \begin{itemize}
         \item[(i)] Suppose $v'_c$ is contained in a labeled proper 4-link $v'_c \xlink v''_c$. In that case, we use a ``trick'' and run the algorithm from the proof of Lemma~\ref{lem:unlabeled-ge4}, as if  $v''_c \link v$ was a proper labeled link  (we can do this since the edges incident to $v''_c$ and $v$ on the $v''_c \link v$ link are labeled; also, note  that the link $v''_c \link v$ contains  at least 5 labeled edges, which guarantees the minimality of the new configuration, obtained after applying the algorithm). This will invert the labels of the edges on the $v''_c \link v$, giving us a configuration where the link $u_c \xlink v'_c$ is labeled and contains the root $e_q$, and where all the labeled links are of length 4, thus concluding the proof. 
        \item[(ii)] Suppose $v'_c$ is not contained in a proper link, it must then contained in a dangling labeled link $v'_c \xlink 
        \bottom$. This dangling labeled link cannot be of length 1 since we know $v$ is deeper than $v'_c$. So it can be of length 2, 3 or 4. 
    In which case $v$ is the vertex endpoint of an unlabeled dangling link $v\link \bottom$, of length respectively 1, 2 or 3. In that case inverting the alternating chain $\{v'_c \xlink          \bottom,v'_c\link  v,v\link \bottom\}$ gives a minimal configuration where all labeled links are of length at most 4, and this concludes the proof.
     \end{itemize} 

\paragraph{Case 2 (dangling labeled link $u_c \xlink \_$):} Suppose the labeled link containing the root is a dangling link of length $\ge 5$. The same argument as above can be adapted in combination with Case 2 in the proof of Lemma~\ref{lem:unlabeled-ge4} to conclude the proof. 

\end{proof}

\section{Proof of Theorem~\ref{thm:4inarow} and~\ref{thm:SBMS}}

\subsection{Pruning}

\begin{definition}\textbf{Pruning}\smallskip

Given a decoding tree $\Tiq$, we define the pruning of $\Tiq$ as a pair  $(\hTiq, \mathcal{S}) $ where:\smallskip

\begin{itemize}
    \item $\mathcal{S}$ is a set  of \textbf{removed subtrees} (each subtree in $\mathcal{S}$ is rooted in some vertex). 
    \item $\hTiq $, the \textbf{pruned tree}, is a connected subgraph of $\Tiq$ containing the root edge $e_q$, obtained by removing the subtrees in $\mathcal{S}$.
\end{itemize}

Since each  subtree  in $\mathcal{S}$ is routed in a vertex, $\hTiq $ has new dangling edges (which are incident in $\Tiq$ to the root of some tree in $\mathcal{S}$).
We refer to those as \textbf{severed edges}.
\end{definition}

\subsection{The Non-Degenerate Error \texorpdfstring{$\epsilon^4$}{epsilon4} of Weight 4 is Undecodable by MS}\label{app:weightfour_a}

We now proceed to show the upper bound on the decoding radius of Theorem~\ref{thm:4inarow}, by showing that the decoder is blind on the syndrome $\synd^4$. 

Let $c$ be an unsatisfied check in $\synd^4$, the syndrome associated to $\epsilon^4$, and $q \in \mathcal{N}(c)$.
Let $\Tiq$ be the decoding tree associated to $q$ at iteration $i$, and let also $\sc$ be the fake syndrome where only $c$ is unsatisfied (Section~\ref{subsec:local-blindness}).  

On  $\Tiq$, we now consider some minimal configurations $\Cx, \Co$ associated to syndrome $\synd$, and  $\Cxc,\Coc$ associated to syndrome $\sc$ where all labeled links are of length $\leq 4$ and for all proper $I$ and $\Lambda$ links, both endpoints are a copy of the same check (thanks to Lemmas~\ref{lem:unlabeled-ge4},\ref{lem:labeled-ge4}, and~\ref{lem:c2c-e4}). 

\paragraph{Outline of the proof:} We now proceed to show that  $|\Cx|-|\Co| = |\Cxc|-|\Coc| $, which implies the blindness property for qubit $q$ by Theorem~\ref{theo:Wiberg96}.

Explore $\Tiq$ by doing a breadth-first search. 
Whenever a vertex $w_{c'}$ associated to a check such that $\synd(w_{c'})=1$  and $c'\neq c$ is found in $\Tiq$, proceed as follows.

Let $g$ be the grandfather of $w_{c'}$ in $\Tiq$, and prune the subtree rooted in $g$ in $\Tiq$. 
After doing that for the whole tree, we are left with a pair $(\hTiq, \mathcal{S})$.

Note that for any subtree in $\mathcal{S}$, it holds that all unsatisfied vertices are at depth~$\ge 2$, as otherwise this would give a path in the toric code between unsatisfied checks of length~$< 4$.

\paragraph{Claim 1:}
In any of the minimal configurations $\Cx, \Co, \Cxc$ and $\Coc$ all severed edges are unlabeled.\smallskip

Consider a severed edge, incident (in $\Tiq$) to the root vertex of some removed subtree $S$.
If $e$ is labeled, it must belong to a labeled link. There are 3 cases:
\begin{itemize}
    \item[(i)] Suppose $e$ belongs to a dangling labeled link $v \xlink \bottom$ with endpoint $v$ in $S$. 
    Since $S$ was removed, $v$ must be at depth $\ge 2$ in $S$. So $v \xlink \bottom$ must be of length $\ge 7$ (2 to go from $v$ to $e$, plus $e$, plus at least $4$ edges to go back down to a dangling edge in $\hTiq$) which is impossible.
    \item[(ii)]  Suppose $e$ belongs to a dangling labeled link $v\xlink \bottom$ with endpoint $v$ in $\hTiq$. Since there is an unsatisfied vertex $u$ at depth 2 in $S$, and that $u$ and $v$ must be at distance $\ge 4$, then $v$ is at distance $\ge 5$ from the bottom, thus the length of $v \xlink \bottom$ is $\ge 5$ which is impossible by assumption. 
    \item[(iii)] Suppose $e$ belongs to a labeled proper link. It must be of the form $u_c \xlink v_c$ where $u_c \in \hTiq$ and $v_c \in S$ and of length 4. 
    But $S$ was removed because of the unsatisfied vertex $w_{c'}$ with $c'\neq c$, which is at depth 2 in $S$. Since $d(u_c,v_c)=4$, $v_c$ is at depth $\ge 2$ in $S$, and  $d(u_c,w_{c'}) \ge 4$, we conclude that $v_c$ must be at depth exactly 2 in $S$, and moreover $d(u_c,v_c)=d(u_c,w_{c'})=4$. Thus, $w_{c'}\link v_c$ is a $\Lambda$-link with both endpoints at distance $4$ from $u_c$, which  yields a contradiction by applying Lemma~\ref{lem:lambda}.
\end{itemize}

\paragraph{Claim 2:} Denoting $\hat \Co, \hat \Cx, \hat \Coc, \hat \Cxc$ the restriction of the minimal configurations to the pruned tree $\hTiq$, we have that $|\hat \Co|=|\hat \Coc|$ and $|\hat \Cx| = |\hat \Cxc|$.

This is because all the severed edges are unlabeled so the labeling induced by configurations $\hat\Cx$ and $\hat \Cxc$ (resp. $\hat\Co$ and $\hat \Coc$) are interchangeable and of same weight, as otherwise it would contradict the minimality of $\Cx$ and $\Cxc$ (resp. $\Co$ and $\Coc$).

\paragraph{Claim 3:} For every subtree $S \in \mathcal{S}$, since the corresponding severed edge $e$ is unlabeled in both $\Co$ and $\Cx$, then the restriction of $\Co$ and $\Cx$ to $S$ have the same weight.

Clearly the minimality of the labeling on such a subtree does not depend on the rest of the tree since $e$ is unlabeled in both cases, so they must have the same weight.\smallskip

By Claims 1 and 3, 
$$|\Cx|-|\Co| = |\hat \Cx|-|\hat \Co| \text{ and } |\Cxc|-|\Coc| = |\hat \Cxc|-|\hat \Coc|$$
By Claim 2, 
$$|\hat \Cx|-|\hat \Co| = |\hat \Cxc|-|\hat \Coc|$$
Hence 
$$|\Cx|-|\Co| = |\Cxc|-|\Coc| $$

\subsection{Decoding of Non-Degenerate Errors of Weight \texorpdfstring{$\le 3$}{le 3} by MS}\label{app:weightfour_b}

We now show that the non-degenerate decoding radius of the MS is $\ge 3$, to conclude the proof of Theorem~\ref{thm:4inarow}. We first give a few definitions and prove two preliminary Lemmas.

Given a toric code $T$, we will refer to a connected planar subgraph $\tau$ induced by one or several plaquettes of $T$ as a patch. Precisely, $\tau$ is a connected subgraph induced by its vertex set $V(\tau)$, which must satisfy  $V(\tau) = \cup_i V(p_i)$, where $V(p_i)$ denotes the set of vertices (checks) of plaquette $p_i$. It is easily seen that if $\tau$ is a patch of $T$ and $T'$ is a toric code of higher minimum distance, then $\tau$ can be embedded in $T'$ via an injective graph homomorphism $\tau \rightarrow T'$, and this embedding is unique up to translations, reflections, and rotations in $T'$. We will use such a patch embedding to embed in $T'$  planar subgraphs of $T$ and syndromes. We note that (i) for any planar subgraph $\sigma$ of $T$ there exists a patch  $\tau$ in $T$ such that $\sigma$ is a subgraph of $\tau$, and (ii) there exists a patch $\tau$ with $V(\tau) = V(T)$. 

\begin{definition}
    Let $T$ and $T'$ be toric codes, with $T'$ of larger minimum distance. 
    \begin{itemize}
       \item[(i)] Consider a planar subgraph $\sigma$ of $T$, it must be contained in  a patch $\tau$ in $T$. We define the embedding of $\sigma$ in $T'$ as the subgraph $\sigma'$ of $T'$ induced by  embedding $\tau$ in $T'$. This embedding is unique up to translations, reflections, and rotations in $T'$.
    
        \item[(ii)] Consider a syndrome $\synd$ on $T$, and a minimal connected (necessarily planar) subgraph $\sigma$  of $T$ covering all its unsatisfied checks. We define the embedding of $\synd$ in $T'$ as the syndrome $\synd'$ on $T'$ induced by  embedding $\sigma$ in $T'$. This embedding is unique up to translations, reflections, and rotations in $T'$. 
    \end{itemize}
\end{definition}

\begin{definition}
    For a toric code $T$, the distance  between two sets of checks $C, C' \subseteq V(T)$ is defined as $d(C,C') := \min \{\,d(c, c') \mid c\in C,\ c' \in C'\,\}$. 
    We extend this definition to qubits, errors, syndromes, and subgraphs of the toric code in the following way (where the same extension applies to the second argument of $d(-, -)$): 
    \begin{itemize}
    \item For a set of qubits $Q\subset E(T)$, $d(Q,-) := d(\mathcal{N}(Q), -)$, where $\mathcal{N}(Q) := \cup_{q\in Q} \mathcal{N}(q)$. 
    \item For a syndrome $\synd$, $d(\synd,-)  :=d(\supp(\synd), -)$,  where $\supp(\synd) := \{c \in V(T) \mid \synd(c)=1\}$ 
    \item For an error $\evec$, $d(\evec,-)  :=d(\supp(\evec), -)$,  where $\supp(\evec) := \{q \in E(T) \mid \evec(q)=1\}$. 
    \item For a subgraph $\sigma$ of $T$, $d(\sigma, -) := d(V(\sigma), -)$, where $V(\sigma)$ is the vertex  set of $\sigma$.
    \end{itemize}
\end{definition}

\begin{definition}
    The \textbf{diameter} of a syndrome $\synd$ is the minimum over the diameters of all the connected subgraphs covering its unsatisfied checks.
\end{definition}

\begin{definition} For a toric code, a 
\textbf{ $\delta$-local syndrome} $\synd$ is a syndrome where all unsatisfied checks are contained in a connected subgraph of diameter $\delta$.
\end{definition}

\begin{definition}
    We say a syndrome $\synd$ is \textbf{decoded in $k$ iterations} if the MS decoder finds an error $\evec$ such that $\Hmat\evec = \synd$ at iteration $k$ and not before.
\end{definition}

\begin{lemma}\label{lem:sigma}
    Let $T$ be a toric code of distance $d$, and $\synd$ be a $\delta$-local syndrome decoded in $k$ iterations.
    If $d \ge 2k+\delta$, then for any toric code $T'$ of distance $d' > d$, the embedding of $\synd$ in $T'$ will be corrected in $k$ iterations. 
\end{lemma}

\begin{proof}
    Let $\sigma$ be the connected subgraph containing all the edges at distance $\le \delta$ from $\synd$, and their check neighborhood, that is:
    $$E(\sigma) = \{e\in E(T) \mid d(e,\synd) \le \delta \} \text{ and } V(\sigma) = \displaystyle\bigcup_{e \in E(\sigma)} \mathcal{N}(e). $$
    Note that $\sigma$ is planar as $\delta < d$ so we can define $\sigma'$ to be the embedding of $\sigma$ in $T'$. 
    For any qubit $q' \in T'$, either $d(q',\synd)\ge k$ in which case one can see that its \aposteriori at each iteration will be $\app(q',i) = 1+2i$, for $i\leq k$ or $d(q',\synd)< k$, in which case, $q' \in \sigma'$, and at each iteration, its \aposteriori is the same as that of $q\in \sigma$, its corresponding qubit in $T$. This is because since $d\ge 2k+\delta$, there are no qubits in $T$ that have a decoding tree reaching unsatisfied checks of $\synd$ by going around the toric code. Formally, this condition can be stated as: given $\tau$ a minimal connected covering of $\synd$, then for all $q$ in $T$, $E(\tau) \cup \{ e \in E(T),\  d(e,q)\le k \}$ does not contain a logical operator.
    Thus at iteration $k$, the decoder will converge on the embedding of $\synd$ on $T'$.
\end{proof}

\begin{lemma}\label{lem:sigma1-sigma2}
    Let $T$ be any toric code, $\synd_1$ a syndrome decoded in $k_1$ iterations and  $\synd_2$ a syndrome decoded in $k_2$ iterations.
    Further assume that $k_1\leq k_2$, and that if the decoder where left to  run on $\synd_1$ for any number $k \in [k_1,k_2] $ of iterations (even after the syndrome is satisfied), the error outputted at each iteration would still satisfy the syndrome.
    Let $\delta$ be the distance between $\synd_1$ and $\synd_2$. 
    If $\delta \ge 2k_2$ then the syndrome $\synd_1 + \synd_2$ is decodable in $k_2$ iterations.
\end{lemma}

\begin{proof}
    For any qubit $q$, either is true:
    \begin{itemize}
        \item[(i)] $d(\synd_1,q) \ge k_2$ and $d(\synd_2,q) \ge k_2$.
        In this case the decoding tree of $q$ does not contain unsatisfied checks for any iteration $k\leq k_2$. So $\app(q,k) = 1+2k$. 
        \item[(ii)] There exist $i,j$ with $\{i,j\} = \{1,2\}$ such that $ d(e_i,q) < k_2 $ and $d(e_{j},q) \ge k_2$
    
        In that case the decoder will take the same decision on $q$ as if the syndrome was only $\synd_i$.
    \end{itemize}
    At iteration $k_2$, by (ii) in the $k_2$-neighborhood of $\synd_1$ and $\synd_2$, both syndromes are satisfied, and by (i) all the other qubits are not in error, so the syndrome $\synd_1 + \synd_2$ is satisfied. 
\end{proof}

We are now ready to finish the proof of Theorem~\ref{thm:4inarow}. We do this by treating the first cases numerically, and then use the Lemmas~\ref{lem:sigma}, \ref{lem:sigma1-sigma2} introduced before to prove the result for $d\ge18$.

\begin{proof}
    
    For any toric code of distance $9 \le d < 18$,  all the non-degenerate errors of weight 
    $\leq 3$ are decodable by the MS algorithm.
    This is checked by exhaustive numerical verification.

For toric codes of minimum distance $\ge 18$, we do the following :

\begin{itemize}
    \item[(i)] It can be easily checked that all weight 1 errors are non-degenerate and can be decoded by the MS in 1 iteration, for any toric code of distance $d\geq 3$.
    \item[(i')] Moreover, it is also true that for a weight 1 error on a toric code of distance $\ge 5$, if the decoder where allowed to continue after the first iteration (when it already decodes the error), at the second iteration it would still output the same weight 1 error. 
    \item[(ii)] Any weight 2 non-degenerate error where the two qubits in error are at distance $\ge 2$ will be decodable in 1 iteration because of (i) and Lemma~\ref{lem:sigma1-sigma2} for any toric code of distance $d\ge 6$.
    \item[(ii')] For the remaining non-degenerate errors of weight 2, those where the distance between the two qubits in error is at most 1 (\idest, of diameter 3), there are only 4 types of error up to symmetries and translations (shown in Fig. \ref{fig:non-degenerate-2}). For those, we check numerically that the decoder converges in at most 2 iterations on a toric code of distance 7. 
    Using Lemma~\ref{lem:sigma}, since $2k+\delta = 2\times 2+3 = 7$, we conclude that all the errors of weight 2 and diameter 3 are corrected on a toric code of any distance $d \ge 7$. 
    \item[(iii)] 
    Any weight 3 error where the distance between any two pair of qubits is $\ge 2$ is correctable in 1 iteration by Lemma \ref{lem:sigma1-sigma2} for any toric code of distance $\ge 3$.
    \item[(iii')]
    Remaining errors can be written as the sum of an error $\evec_1$ of weight 1 (of diameter~1) and an error $\evec_2$ of weight~2 and of diameter~$\le 3$ (depicted in Figure~\ref{fig:non-degenerate-2}). 
    Suppose $d(\evec_1,\evec_2)\ge 4$.
    By (i), $\evec_1$  is decoded in $1$  iterations.
    and by (ii'), $\evec_2$ is decoded in  $\le 2$  iterations.
    We can thus use Lemma~\ref{lem:sigma1-sigma2} on the syndromes associated to $\evec_1$ and $\evec_2$ to conclude that $\evec_1+\evec_2$ is decoded in 2 iterations. Note that this argument is true for any error of weight 3 satisfying the condition $d(\evec_1,\evec_2)\ge 4$ irrespective of the distance of the code.
    \item[(iii'')] For the remaining finite number of weight 3 errors of diameter $\delta \le 8=4+1+3$ (where $d(\evec_2,\evec_2) \leq 4$, the diameter of $\evec_1$ is 1 and the diameter of $\evec_2$ is at most 3). We check numerically that these errors are decodable in $\leq 5$ iterations on a toric code of distance $18$. By applying Lemma~\ref{lem:sigma}, and since $2k+\delta = 2\times 5 + 8 \le 18$, we conclude that for any toric code of distance $d\ge 18$, all weight $3$ non-degenerate errors are decodable.
\end{itemize}
\end{proof}

\begin{figure}
\center
\includegraphics{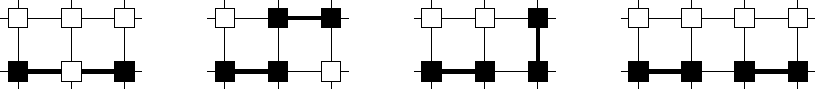}
\caption{The weight 2 errors of diameter 3 (up to symmetries and translation), all are decodable in 2 iterations of the MS.
}\label{fig:non-degenerate-2}
\end{figure}

\subsection{Proof of Theorem~\ref{thm:SBMS}}\label{app:sbms-proof}

We now prove Theorem~\ref{thm:SBMS}, using the tools introduced in the proof of the second part of Theorem~\ref{thm:4inarow}.

For all toric codes of distance $ 7 \le d \le 28$ we enumerate all errors of weight $\le 3$ and check that SB+MS is able to correct them.
There are two cases:
\begin{itemize}
    \item[(i)] the diameter of the error is $\leq 8$.
    \item[(ii)] there must be one qubit in error that is at distance $\ge 4$ from the others.
\end{itemize}

For case (i), we know that all errors of weight $\le 3$ take at most 10 iterations to converge, hence using Lemma~\ref{lem:sigma}, we conclude that this holds for any toric code of distance $\ge 2\times 10+8=28$.
Case (ii) can be taken care of using Lemma~\ref{lem:sigma1-sigma2} using the same argument as above, so we conclude this is true for any toric code of distance $\ge 28$.

\bibliographystyle{IEEEtran}
\bibliography{biblio.bib}

\end{document}